\pdfoutput=1
\documentclass{article}

\usepackage[preprint,nonatbib]{neurips_2020}

\usepackage[utf8]{inputenc} 
\usepackage[T1]{fontenc} 
\usepackage{hyperref} 
\usepackage{url} 
\usepackage{booktabs} 
\usepackage{amsfonts} 
\usepackage{nicefrac} 
\usepackage{microtype} 
\usepackage{graphicx}
\usepackage{algorithm}
\usepackage{amsmath}
\usepackage{amsthm}
\usepackage{amssymb}
\usepackage{algpseudocode}
\usepackage{xcolor}
\usepackage{subfigure}
\usepackage{float}
\usepackage{hyperref}

\def \E{{\rm E}}

\def \min{{\rm min}}

\theoremstyle{plain}
\newtheorem{theorem}{Theorem}

\theoremstyle{definition}

\newtheorem{remark}[theorem]{Remark}

\DeclareMathOperator*{\argmax}{\arg\!\max}
\newcommand{\expnumber}[2]{{#1}\mathrm{e}{#2}}
\DeclareSymbolFont{symbolsC}{U}{pxsyc}{m}{n}
\DeclareMathSymbol{\coloneqq}{\mathrel}{symbolsC}{"42}

\title{A Comparison of Reinforcement Learning and Deep Trajectory Based Stochastic Control Agents for Stepwise Mean-Variance Hedging}

\author{
Ali Fathi \thanks{All contents and opinions expressed in this document are solely those of the authors and do
not represent the view of Wells Fargo Bank NA.} \\
Corporate Model Risk\\
Wells Fargo\\
\And
Bernhard Hientzsch \footnotemark[1] \\
Corporate Model Risk\\
Wells Fargo\\
}

\begin{document}

\maketitle

\begin{abstract}

We consider two data-driven approaches to hedging, Reinforcement Learning and Deep Trajectory-based Stochastic Optimal Control, under a stepwise mean-variance objective. 
We compare their performance for a European call option in the presence of transaction costs under discrete trading schedules. 
We do this for a setting where stock prices follow Black-Scholes-Merton dynamics 
and the "book-keeping" price for the option is given by the Black-Scholes-Merton model with the same parameters. 
This simulated data setting provides a "sanitized" lab environment with simple enough features 
where we can conduct a detailed study of strengths, features, issues, and limitations of these two approaches. 
However, the formulation is model free and could allow any other setting with available book-keeping prices. 
We consider this study as a first step to develop, test, and validate autonomous hedging agents, 
and we provide blueprints for such efforts that address various concerns and requirements. 
\end{abstract}

\section{Introduction}

Recent advances in data driven hedging have been heralded as one of the most exciting applications of machine learning (ML) and in particular reinforcement learning (RL) in derivatives pricing and risk management (\cite{deep_Bellman}, \cite{Kolm_Ritter}, also see \cite{Risknet}). 
One of the reasons for this excitement is that it is relatively easy to:
(1) simulate or generate data for the behavior of hedging instruments, 
(2) model the value of a trading strategy over time given such data 
-- even for realistically modeled incomplete markets where trades cause friction and attract transaction costs,
and (3) compute hedging or risk-management objectives of such strategies.
The hedging objectives could be based on how the strategies perform as measured against 
book-keeping prices for to-be-hedged portfolios or against to-be-made (or to-be-received) payoffs for these to-be-hedged portfolios. 
This means that it is straightforward to model the behavior of the trading strategies even as consistent prices or sensitivities 
might be difficult to compute in conventional models. 
Therefore one can train data-driven agents that optimize the hedging and risk-management objectives 
and perform better than conventional sensitivity-based and 
model-based hedging approaches that are currently often used for such hedging.  
Data-driven agents can also take into account  other factors and automate hedging even in cases where 
conventional hedging approaches require ad-hoc decisions and tweaks by traders and risk managers 
to address the shortcomings of the existing approaches. 

Several banks and other financial institutions have started to apply these new approaches to hedging and risk management (see \cite{RisknetJPM} for instance). 
However, in our experience, one of the main concerns the quants or the traders express regarding the data-driven hedging agents is the 
\textit{model risk} aspects of deploying such solutions \cite{SR11-7}, and there seems to be a need for a framework to manage the model risk which may arise from using them. 
To contribute to wider adoption of data-driven hedging solutions across the financial industry, we are studying here how we can apply 
RL and Deep Trajectory-based Stochastic Optimal Control (DTSOC) approaches to hedging problems 
(in particular the stepwise mean-variance hedging setting), what are commonalities and differences,
what are advantages, features, issues, and limitations. We also opine on how such issues and limitation may be addressed. 

\subsection{Scope of This Work}

This paper concentrates on 
hedging strategies that optimize the mean and variance of hedged portfolios (containing hedged instruments and hedging instruments) by optimizing a cumulative stepwise 
mean-variance objective, but the setting will apply to many trading strategies with different objectives (and we will discuss other settings in future papers). 

To define the set up, we will need to model how the prices of the hedging and the hedged instruments evolve, 
how the trading strategy to do so performs (including transaction costs), 
and define the stepwise objectives to be optimized over in terms of the prices, costs, etc. involved in the strategy. 

In mean-variance hedging, one assumes that both the prices of the hedging and the hedged instruments to define the hedging objectives are known. This might be the case 
if one tries to learn hedging strategies from settings in which both hedged instruments (say, vanilla options) and hedging instruments (underlier and cash) are
liquidly quoted or where one has a book-keeping model for the hedged instruments that will provide the price. The first setting might make sense if one has over the counter
(OTC) or exchange-traded options on the books and needs to hedge them with the underlying even though at least some such options are liquidly traded. The second 
setting might make sense for a trader in an investment bank that hedges an option on the books which will be valued by that book-keeping model\footnote{
We call a model a book-keeping model if it reflects the price or value of an instrument in some way but does not necessarily take all details and features 
of the instrument and/or the market into account. For instance, Black-Scholes type models and similar "simpler" models are often used to mark instruments 
in systems of record (SOR) and they are validated and approved for such use because even though they do not reflect all features and details, they reflect 
value well enough for purpose. In particular, if the model is contained in a system of record, we call it a SOR model. Sometimes, for purposes of margins or 
collateral or trading through central counterparties, using such simpler book-keeping models has the advantage that there will be less uncertainty 
and fewer arguments about the model parameters and settings, compared to more complicated models.} in the system of record (SOR) 
and the trader will be judged by their performance against the prices contained in the SOR. 

There might be settings in 
which neither the hedged instrument nor any closely related instrument are liquidly quoted or in which there are no appropriate and flexible enough book-keeping 
models available. In these circumstances, one would need stepwise or global objectives that are defined without reference to (book-keeping or other) prices of
the to-be-hedged instruments. These could be for instance objectives that consider how well the trading strategy replicates or risk-manages all cashflows of 
the to-be-hedged instruments, for instance in a squared difference sense, leading to \textit{quadratic hedging} (we leave this case to a future paper).

We consider a case in which we have a model for how appropriate risk factors and prices of the hedging and the hedged instruments 
evolve that are needed to compute the objectives. What we mean by that is that we have some implementation that is able to generate as many trajectories 
of these instrument prices and other risk factors as needed. Based on these trajectories, we will be able to simulate the trading strategy and their performance 
and compute the appropriate objectives. We will here assume that the trading from our trading strategies does not impact how the prices of the hedging and the hedged instruments
evolve, just how the trading strategy and the objectives evolve, which is a common assumptions in current works. We do not specify how this model is obtained, 
setup, or run; only that we can obtain those trajectories as needed. As such, the model could be a well-specified traditional model or it could be 
a trained generative model (\cite{boursin2022deep}, \cite{Cohen_neural_SDE}). Here, we will work with simulated data from Black-Scholes model and defer the settings involving synthetic data from generative models to future work. 

We will test and validate our agents on simulated data from such models and leave the question whether such models trained on simulated or generated data 
can perform well on observed data for our settings to future work. There are indications in the literature that models trained on 
appropriate simulated or synthetic data can perform well when applied to observed data from the market (as in \cite{boursin2022deep}). 

We are not using the model(s) or any information about the model(s) in the learning and training beyond using it to generate trajectories and data including
stepwise objectives. In this way, we are operating in a model-free, data-driven setting.\footnote{Particular setups might include state variables and/or compute features
that are based on a particular model or setting, such as an option Delta computed under the assumption that the dynamics of the hedging instrument can be approximated
by a Black-Scholes type model, but the algorithms to learn trading strategies or value functions do not use or require any model details beyond what is contained 
in the state and feature space.} 
Given the above discussion, we are actually operating in a rich-data 
(even infinite-data) regime, corresponding to a well-defined conventional or generative model. 

The rest of the paper is structured as follows: First, we will discuss the instrument and trading strategy setup and mean-variance hedging. 
Then, we will discuss how Reinforcement Learning (RL) conceptualizes the problem and the setting,
introduce the RL techniques and algorithms that we will apply,
and then introduce the deep trajectory-based empirical stochastic optimal control (DTSOC) 
approaches and algorithm. 
We will then specify the experimental setup and the models, report and analyze outcomes,
visualize and interpret outcomes, and analyze sensitivity of the results with respect to 
various choices. 
Finally, we conclude and discuss some possible future directions. 

\subsection{Related Work}
We discuss here some of the existing work in the literature on the use of RL and neural networks for pricing and hedging options (see \cite{RLfinance} for a broader survey of RL applications in finance).

In \cite{QLBS}, a European call option in the Black-Scholes model (trading at discrete times, no transaction costs) 
is hedged and priced using a $Q$-learning method (there, called QLBS method). 
The objective is minimizing the sum of the initial cash position and the weighted discounted sum of variances of the 
hedging portfolio at all subsequent time steps. This setting and objective 
can be identified with a corresponding MDP setup 
and thus solved with standard RL/DP approaches such as $Q$-learning \cite{QLBS}. 
The price at maturity has to be equal to the final payoff and prices respective $Q$-functions at earlier times
are determined by Dynamic Programming or other standard RL approaches. 

 The  mean-variance hedging problem  in a setting which includes transaction costs is considered in \cite{cao2021deep}. 
However, they define the reward objective in global fashion - it is the expected total cost for the entire hedging and its variance.
They use versions of $Q$-learning and other methods to address this problem. 
They do not treat stepwise mean-variance hedging. 

 The stepwise mean-variance hedging under transaction costs for a single option is considered in \cite{Kolm_Ritter}. 
 The hedging agent is trained with an RL algorithm which seems to be a version of $Q$-learning. 
 Also, stepwise mean-variance hedging under transaction costs with various RL algorithms is considered in \cite{du2020deep}. 

In \cite{mikkila2021empirical}, authors consider a case where the stock price\footnote{Actually, the authors are using the value of the underlying index, 
and are assuming that the underlying index is tradeable, which is not completely accurate (futures on the index are traded but the spot index, 
in particular through changes to the index, is not tradeable and not entirely replicable through self-financing strategies).} 
and the option price are taken from higher frequency data acquired from an option exchange and used to train an agent with an RL algorithm. 
They use a stepwise hedging objective; however, it is defined in terms of mean and standard deviation 
(and in terms of P\&L and absolute value of P\&L for the sample version). If the option is quoted and traded on an exchange, 
it might be less necessary to hedge or replicate said option. 

There exists a large literature on applications of neural networks for option pricing. 
Most of these works consider the neural networks as a numerical tool to accelerate calculation tasks (such as model calibration \cite{hernandez2016model}, 
solving PDEs \cite{Deep_PDE}, etc.) more efficiently and in higher dimensions (also see \cite{ruf2019neural} for a literature review). 
Among these, we distinguish the line of works initiated in \cite{han2016deepsc} 
where neural networks are used to solve the stochastic optimal control problem numerically.
More concretely, the DTSOC method in \cite{han2016deepsc} is a setup for trajectory based empirical deep stochastic optimal control with both stepwise and final objectives/costs. 
Finally, the "deep hedging" paper \cite{Deep_Hedging} presents a trajectory based empirical deep stochastic optimal control
 approach to minimizing some global objectives related to replication and/or risk management of some final payoff, 
 however, it does not consider a stepwise objective, just a global/final objective.

\section{Framework}

\subsection{General Setup}

We consider here dynamic trading, hedging, and/or risk management problems. These  problems are typically posed on some universe of given instruments (including cash/bank account) and the trader or agent has to select a strategy that determines  how much of each instrument to hold at each time to
achieve a certain trading, hedging, or risk management objective. 
As is common in such settings, we assume that there is a given fixed grid of times
 at which trades can occur and that such trades can incur transaction costs. We also assume that the strategies are 
self-financing. 
The objectives can be defined at a certain time horizon or they could be defined for each time period between trading times (stepwise objectives).
Here, we will consider a stepwise objective. 
These stepwise objectives do not necessarily have to directly correspond to trading gain or loss or transaction cost or any financial item, they can also be defined in terms of moments or other properties of
the distribution of specific gains, losses, costs, or appropriate combination thereof.

On a basic level, this is a stochastic sequential decision problem. 
The prices of the given instruments will follow some stochastic process and it is assumed that we are given a model or generator for this process. 
Based on those prices and any other appropriate risk factors  that can be computed based on the history of said prices, 
we need to decide at each trading time-step
what the new holdings of the given instruments will be and thus what the necessary trade sizes are 
(so holdings and/or trade sizes are the appropriate decision variables), taking into account the transaction costs and other 
impacts associated with that trade. 
Given those holdings and cash-at-hand at the beginning of the no-trade interval, the value of all those holdings and cash at the end of the no-trade interval follows 
from the stochastic evolution of the prices of the given instruments and interest on the cash account. From this, one can define the stepwise objective, which often is based on an appropriately defined 
change in strategy value over that period. 
We assume here that we have an implementation or model that describes the evolution of the trading strategy and also the computation of the stepwise objectives.

\subsection{Trading Strategy Corresponding to Hedged Portfolio} 
We consider the case of a hedged portfolio. This means that we are given a fixed holding $H^O$ of the to-be-hedged portfolio.\footnote{Which could be easily extended to 
a priori given holding level $H^O_t$ over time.}
It is assumed that the price of the to-be-hedged portfolio is given as $O_t$, corresponding to a quoted/observed or a computable book-keeping price. We do not rely on how this price is computed, 
only that we are provided with some implementation of it. It is also assumed that we are given other "hedging" instruments also and we denote their prices by $S_t$ (in general a vector) and we 
denote the holding of them in our strategy at time $t$ with $H^S_t$ (a vector of corresponding length).

 For our examples, we will typically assume a single hedging instrument. 
Finally, we assume that excess cash from the trading strategy is deposited in a money-market account with a certain given interest accrual and discount factor (and that negative cash 
balances can be borrowed at the same interest accrual and discount factor). We assume that we initially start with only the fixed holding $H^O$ and some initial cash $Y_0$ but no holdings
in the hedging instruments (i.e., $H^S_{0^{-}}=0$). The trading times are given as $t_i$ with $t_0=0$. If referred to, $t_{-1}$ is a time just before time 0 with no holding in 
hedging instruments. 

Consider time $t_{i+1}$, right before the trade: The value of the hedging part of the strategy then is  
\begin{equation}
Y_{t_{i+1}} = \frac{\left( Y_{t_{i}} - H^S_{t_i} S_{t_i} - \mathsf{TC}(H^S_{t_i},H^S_{t_{i-1}},S_{t_i}) \right)} {\mathsf{DF}_{i,i+1}} + H^S_{t_i} S_{t_{i+1}} 
\end{equation} 

This reflects that at the last trade time $t_i$, the portfolio was rebalanced to the appropriate holdings $H^S_{t_i}$, transaction costs were charged 
$\mathsf{TC}(H^S_{t_i},H^S_{t_{i-1}},S_{t_i})$, the thus resulting cash balance $Y_{t_{i}} - H^S_{t_i} S_{t_i} - \mathsf{TC}(H^S_{t_i},H^S_{t_{i-1}},S_{t_i})$ attract stepwise
interest (interest accrual corresponds to dividing by appropriate stepwise discount factor $DF_{i,i+1}$). At the same time, the holdings in the hedging instruments are now valued based 
on the new price of the hedging instruments. 

It helps to rewrite the equation as follows:
\begin{equation}
Y_{t_{i+1}} = \frac {1} {\mathsf{DF}_{i,i+1}} Y_{t_i} + H^S_{t_i} \left( S_{t_{i+1}} - \frac {S_{t_i}} {\mathsf{DF}_{i,i+1}} \right)
- \frac{\mathsf{TC}(H^S_{t_i},H^S_{t_{i-1}},S_{t_i})} {\mathsf{DF}_{i,i+1}} 
\end{equation} 

In this completely linear setting without stochastic interest rates (and assuming that the transaction cost is linear in the instrument prices for positive multipliers $c$,
$ \mathsf{TC}(\cdot,\cdot,cS) = c \mathsf{TC}(\cdot,\cdot,S)$), 
it is possible to "hide" the discounting in the value definitions in the following sense: with $\tilde{Y}_{t_i} = \mathsf{DF}_{0,i} Y_{t_i}$ and $\tilde{S}_{t_i} = \mathsf{DF}_{0,i} S_{t_i}$,
we have
\begin{equation}
\tilde{Y}_{t_{i+1}} = \tilde{Y}_{t_i} + H^S_{t_i} \left( \tilde{S}_{t_{i+1}} - \tilde{S}_{t_i} \right)
- \mathsf{TC}(H^S_{t_i},H^S_{t_{i-1}},\tilde{S}_{t_i})  
\end{equation} 

In particular, this means for the increment
\begin{equation}
\tilde{dY}_{t_{i+1}} = H^S_{t_i} \left( \tilde{S}_{t_{i+1}} - \tilde{S}_{t_i} \right)
- \mathsf{TC}(H^S_{t_i},H^S_{t_{i-1}},\tilde{S}_{t_i}) 
\end{equation} 
and we thus do not get any contribution or impact of the initial cash position $Y_0$ beyond an additive shift.

If we are looking at the change of value in the total portfolio (both hedged and hedging instruments), we obtain
\begin{equation}
\delta \tilde{V}_i = \tilde{\mathsf{PnL}}_i =H^O \left( \tilde{O}_{t_{i+1}} - \tilde{O}_{t_i} \right) + H^S_{t_i} \left( \tilde{S}_{t_{i+1}} - \tilde{S}_{t_i} \right)
- \mathsf{TC}(H^S_{t_i},H^S_{t_{i-1}},\tilde{S}_{t_i})  
\end{equation} 
where we denoted $V_t= H^O O_t + Y_t$ and  $\tilde{V}_t= H^O \tilde{O}_t + \tilde{Y}_t$  . 
In the following, we will work with values discounted back to initial time and will omit the tilde over the variables.

Extensions to stochastic interest rates and discounting, differential rates, more involved funding policies are possible, 
require further notation and details, and can be handled similarly, but are not necessary for the settings we discuss here. 
Differential rates and more involved funding policies could be included in the framework by adding bank accounts for positive 
and negative balances and other funding instruments to the hedging instruments and adding constraints to the decision variables (only positive holdings of positive bank account, only 
negative holdings of negative bank account i.e. bank loan, restricting secured loans/repos to the amount held in the corresponding collateral, etc.).

\subsection{Stepwise Mean-Variance Objective}

Often, one would like to control both the mean and variance of the stepwise gains or losses by maximizing mean of gains penalized by a multiple of variance 
(or equivalent minimizing mean of loss combined with variance). When sampling, one can approximate the variance by the second moment (see \cite{ritter2017machine} for a discussion) which means that 
one can express that objective by the expectation of costs
\begin{equation}
c_i = - \mathsf{PnL}_i+ \frac{\lambda}{2} (\mathsf{PnL}_i)^2.
\end{equation}
This can now be used as a cost or reward in RL or DTSOC algorithms. 
Costs to unwind the position at maturity (or convert it into the position to match the payoff of the hedged instrument) can 
be included as final cost if needed. 

\subsection{Relation to Global Mean-Variance Objective}

In some settings, one tries to control only the mean and variance globally, i.e. a mean and variance combination of 
only the final wealth 
is optimized over:
\begin{equation*}
\min_{\text{strategies}}\E[-V_T]+\frac{\lambda}{2}\mathbb{V}[V_T]
\end{equation*}
where $V_t$ is the wealth/value process for the overall hedged portfolio as derived earlier in this section (see, for instance \cite{Kolm_Ritter}).

Using the following discrete time decompositions for the mean and variance of the terminal wealth under appropriate assumptions for the variance, 
for instance a random walk assumption as in \cite{ritter2017machine},
\begin{equation*}
\E[V_T] = V_0+\sum_{1}^T\E[\delta V_{t}].
\end{equation*}
and
\begin{equation*}
\mathbb{V}[V_T] = \sum_{t=1}^T\mathbb{V}[\delta V_{t}],
\end{equation*}
\cite{Kolm_Ritter} rewrite the mean variance hedging (MVH) objective in the following way,
\begin{equation}\label{MVH_Kolm_Ritter}
\min_{\text{strategies}}\sum_{t=0}^{T}\left(\E[-\delta V_t]+\frac{\lambda}{2}\mathbb{V}[\delta V_t]\right)
\end{equation}
and thus obtain a stepwise objective as the one discussed in the last subsection. 

\subsection{Modeling the Hedging and Hedged Instruments} 

The formulations of costs/rewards in the above sections only depend on discounted prices of the hedging instruments, hedged instrument, and transaction costs; 
but do not need any further information about the models for these instruments and prices. With the costs/rewards given above, the agent only 
needs to determine how many units to hold in each of the hedging instruments after rebalancing at each rebalancing time $t_i$, i.e. $H^S_{t_i}$.
Alternatively, one could specify the units to be sold or bought at any given time together with the initial position; or assume that the 
amount sold or bought is proportional to the time between trades or some other variable and then specify that proportional trading rate
in that setting. 

For the hedged instrument, one needs to know the position in the instrument $H^O$ (i.e., whether long or short or some particular holdings
at particular times) and additional information about the instrument - at least time to maturity and/or time. 
The minimal state for modeling this decision problem are these discounted prices (hedging and hedged instruments), 
the amount of the hedging instruments, time, and whatever state the models for these prices and the transaction costs need. 

As for the models for the prices of the hedging and hedged instrument, one can model them under pricing or observational measure,
with some differences in parameters and calibration or fitting. For the hedging instruments, one can use conventional quantitative 
finance models, possibly with hidden factors, such as Black-Scholes, Local Volatility Model, Heston model (with stochastic variance), 
quadratic rough Heston model, etc. One could use generative models such as GANs or appropriately trained neural Stochastic Differential Equations (SDE)s or similar. Finally, 
one could use historically observed data over one instrument or a cross section of similar instruments together with appropriate assumptions 
to generate possible future price movements or future prices.\footnote{See also \cite{cohen2021blackbox} for a discussion of possible 
approaches together with the associated "model" risks.}    

For the hedged instrument, one of the challenges is that the prices of the hedged and the prices of the underlying instruments  need to be consistent in 
an appropriate sense because otherwise they might allow arbitrage strategies. One could implement some book-keeping model which would give the price of that hedged instrument in terms of 
state variables and other parameters. In a trading setup in an investment bank or a hedge fund, holdings will be marked by some 
model within a system of record, and the book-keeping model would be that model in the system of record. For faster computation, 
one could try to learn a surrogate that approximately replicates that SOR model but can be run much faster. Alternatively, if the reference prices for the hedged instrument are observable, one may directly model such prices either as dependent on 
underlying prices or jointly with underlying prices, with appropriate conditions. This, however,
will not work for hedged instruments that are illiquidly or rarely traded or for which no consistent pricing is known. 

One could try to determine a consistent model according to some conceptualization, such as replication with controlled risk or market making 
according to certain strategies and models. However, the construction and validation of such consistent models is very hard, and under such 
circumstances, one most probably should design trading and risk management strategies that describe and manage the risk without requiring
book-keeping and/or consistent modeling, such as hedging that tries to control the replication and risk management of the payoffs of the 
to-be-hedged instrument rather than targets a certain book-keeping value that does not adequately reflect the nature and risk of the instrument. 

Particularly simple examples concern the hedging of a short call option $H_O=-1$ that we sold to some counterparty or the 
hedging of a long call option $H_O=1$ that we bought, use the Black-Scholes model with constant (or time-dependent parameters), with log-Euler time discretization (which 
is exact assuming the parameters are properly chosen), with the Black-Scholes Formula for a call as book-keeping model,
and size of holding as action and decision variable. 

In formulas, this would mean for the price of the hedging instrument
\begin{equation}
S_{t_{i+1}} = S_{t_i} \exp \left( (\mu - \frac{1}{2} \sigma^2) \Delta t_i + \sigma \Delta W_i \right)
\end{equation}
or
\begin{equation}
\tilde{S}_{t_{i+1}} = \tilde{S}_{t_i} \exp \left( (\mu - r - \frac{1}{2} \sigma^2) \Delta t_i + \sigma \Delta W_i \right),
\end{equation}
such that $\tilde{S}_{t_{i+1}}= f_1(\tilde{S}_{t_i},\Delta W_i)$ for an appropriate function $f_1$ and a standard normal $\Delta W_i$
with covariance $\Delta t_i$.
For the price of the call option, we use the Black-Scholes formula as book-keeping model. 
\begin{equation}
\tilde{C}_{t_{i+1}} = \exp(-r t_i) \textsf{BSCall}(S_{t_i},t_i,\mu,\sigma,r)
\end{equation}
The minimal state would consist of $s_i=(t_i,\tilde{S}_{t_i},\tilde{C}_{t_i},H^S_{t_i})$ with the last element being the action (or
being impacted by the action). One can also add other features and information to the state space that the agent could potentially use
to make better decisions or that allow the agent to be learned more easily. 

If one considers more complicated models based on conventional quantitative finance models, one would simulate the (discounted) underlying instruments 
with that model and one would replace the Black-Scholes formula with an appropriate pricing formula or pricer under that model (an 
appropriate future value computation under the book-keeping model). If these models have additional factors, these factors would be
added to the state. If some of these factors are latent or hidden factors, one would need to add some mechanism how these latent 
factors can be estimated or taken into account by some process on observed quantities, add these observed quantities to the 
state, and learn agents that only depend on observed and observable quantities, not the latent factors that will in general 
be unknown (and unknowable) to the agent.  

In general, trading in hedging instruments could impact the prices in those hedging instruments either temporarily or permanently. 
We assume here that the hedging instruments are traded liquidly and that the hedged instrument is such that it can be hedged 
without impacting the prices of the hedging instruments. To a certain extent, short-term price impact can be modeled by and 
absorbed into the transaction cost terms.  

As discussed in abstract and introduction, here we focus on Black-Scholes type models to investigate the agents 
and algorithms in a setting where the model and the features are simple enough, and will consider other models in future work.

\section{Reinforcement Learning}

Reinforcement Learning (RL) is a framework for solving problems consisting of a learning agent interacting with an environment according to a decision policy. The interaction of the agent with the environment is accompanied by receiving rewards (perhaps with delay) and the agent must learn to act according to an optimal policy which maximizes the expected sum of reward received over a given time horizon. The environment consists of everything external to the agent, and is the source of the agent’s observations as well as the reward.

RL framework can be effectively applied to an important class of sequential decision-making problems which can be recursively decomposed into sub-problems, where the result of taking a particular action does
not depend on the prior history of the system up to that point. We first start by formalizing the setup for those problems.

\subsection{Markov Decision Processes}
We begin by defining a Markov decision process (MDP). A discounted finite horizon MDP is defined by a tuple $(\mathcal{S},\mathcal{A}, P, R, \gamma, T)$, where $\mathcal{S}$ is the set states, $\mathcal{A}$ is the set of all actions, $P: \mathcal{S}\times \mathcal{A}\rightarrow \mathcal{P}(\mathcal{S})$ is the transition probability density , $R:\mathcal{S}\times\mathcal{A}\times\mathcal{S} \rightarrow \mathcal{R}$ is the immediate reward (could also be stochastic), $\gamma\in(0, 1]$ is the reward discount factor and $T$ is the time horizon ($T<\infty$). By taking any action $a\in\mathcal{A}$ at the state $s\in\mathcal{S}$, $P(.| s, a)$ defines the probability distribution of the next state and $R(.| s, a,s')$, the distribution of the immediate reward.

A policy $\pi:\mathcal{S}\rightarrow\mathcal{P}(\mathcal{A})$ maps a state $s\in\mathcal{S}$ to a probability distribution $\pi(.|s,a)$
over the set of actions $\mathcal{A}$. The interaction of an agent with the environment is formalized as follows, following a policy $\pi$, starting from the state $s_0$, at each time step $t$, the agent observes the state $s_t$ and takes an action $a_t$ according to the policy $a_t \sim \pi(s_t)$. When the agent performs action $a_t$, the environment makes a stochastic state transition to a new state $s_{t+1}$ according to the probability distribution $P(s_{t+1}|s_t, a_t)$. The agent receives a reward $r(s_t, u_t)$ drawn from the distribution $R(r|s,a,s')$. This reward is a random variable, because it depends on all of the stochastic
transitions up to time $t$. 
The agent is allowed to make decisions, however, the transition density $ P$ and the reward distribution $R$ are dictated by the environment. 
While the agent will observe transitions and rewards, the underlying model and details for $P$ and $R$ are not known to the agent in general.

Therefore, the evolution of the MDP following the policy $\pi$ (also called the trajectory) is given by,
\begin{equation*}
\left\{(a_t,r_t, s_{t+1}), a_t\sim\pi(.|s_t), s_{t+1}\sim P(.|s_t,a_t), r_t\sim R(.|s_t,a_t,s_{t+1})\right\}, t=0,1,\cdots,T.
\end{equation*}

The value of a state for a given policy $\pi$ is the cumulative reward across the trajectories starting from $s_0=s$ and following $\pi$ along the way. The \textit{value function} corresponding to the policy $\pi$ is defined as,
\begin{equation}
V^{\pi}(s) =\E\left[\sum_{t=0}^T\gamma^t r_t\middle|s_0=s\right].
\end{equation}

The action-value function or the \textit{Q-function} at state $s$ and action $a$ is the value of taking the action $a$ and following the policy $\pi$ onwards,

\begin{equation}
Q^{\pi}(s,a) = \E\left[\sum_{t=0}^T\gamma^t r_t\middle|s_0=s, a_0=a\right].
\end{equation}
The value function can be expressed as the expectation of the Q-function across all possible actions,

\begin{equation}
V^{\pi}(s) = \E\left[Q^{\pi}(s,a)\middle|a_t\sim\pi(.|s)\right].
\end{equation}

The Q-function satisfies the Bellman expectation equation,
\begin{equation}\label{Q_Bellman_equation}
Q^{\pi}(s,a) = \E\left[ r_{t} + \gamma Q^{\pi}(s', a') \;\middle|\; s_t = s, a_t = a\right].
\end{equation}
The optimal policy $\pi^\ast$ gives the maximum action-value function $Q^{\ast}_t(s,a)$ for any $s$ and $a$,
\begin{equation}
Q^{\ast}(s,a) \ge Q^{\pi}(s,a).
\end{equation}
The optimal Q-function satisfies the \textit{Bellman optimality} equation (see \cite{sutton2018reinforcement}):
\begin{equation}
Q^{\ast}(s,a) = \E\left[ r_{t} + \gamma \max_{a'} Q^{\ast}(s', a') \;\middle|\; s_t = s, a_t = a\right]. \label{bellman_optimality_equation}
\end{equation}

\subsection{Solving MDPs with Q-Learning}

The tabular Q-learning algorithm finds the optimal Q-function (and hence the optimal policy $\pi^\ast$) satisfying (\ref{bellman_optimality_equation}). To formalize the notion of optimal policy, for each Q-function, we define the \textit{greedy policy} $\pi^Q$ as the policy that chooses the action maximizing the Q-function,
\begin{equation}\label{greedy-policy}
\pi^Q(a|s)=0~~\text{if} ~~Q(s,a)\neq\max_{a'\in\mathcal{A}}Q(s,a'). 
\end{equation}
The optimal policy $\pi^*$ is defined as greedy policy for the optimal Q-function $Q^*$\footnote{If there is more than one maximizer, the greedy policy can be defined as a uniform distribution over such maximizers.}.

In order to find the optimal $Q^*$-function, define the Bellman optimality operator $T$ as follows,

\begin{equation}
TQ(s,a) = r(s,a)+ \gamma\E\left[\max_{a'} Q(s', a') \;\middle|s'\sim P(.|s,a) \right]. \label{bellman_iteration}
\end{equation}

It is observed that equation \eqref{bellman_optimality_equation} describes a fixed point $TQ^*=Q^*$ of the Bellman operator in a suitable function space. 
In fact, it can be shown that the Bellman operator is a $\gamma$-contraction and hence admits a fixed point\footnote{The proof in the tabular case requires $\gamma\in[0,1)$.}. 
The \textit{Q-value iteration} algorithm approximates this fixed point by initializing a random $Q_0$ and 
constructing a sequence of action-value functions $\{Q_k\}$ by defining $Q_k = TQ_{k-1}$. 
It is shown that $\{Q_k\}$ converges to the optimal value function at a linear rate 
(see \cite{RL_book2}, \cite{sutton2018reinforcement}, also \cite{DQL_mathematical}).

\begin{algorithm}
\caption{Q-value Iteration}\label{Q_value_iteration}
\begin{algorithmic}
\State Initialize an arbitrary $Q_0\in\mathbb{R}^{(|\mathcal{S}|\times|\mathcal{A}|)}$
\State At every iteration $k$, improve the $Q$-value as,
\begin{equation}\label{Q-update}
Q^k(s,a)=r(s,a)+\gamma\E_{s'}\left[\max_{a'}Q^{k-1}(s',a')|s,a\right]
\end{equation}
\State Stop whenever, $||Q^k-Q^{k-1}||_{\infty}$ is small.
\end{algorithmic}
\end{algorithm}

The \textit{Q-learning (tabular)} algorithm approximates the updates in equation (\ref{Q-update}) by replacing the expectation with sample observations. At any step $t=1,2,\cdots,T-1$, after taking some action $a_t$, the algorithm observes the immediate reward $r_t$ and the state transition $s_{t+1}\sim P(.,s_t,a_t)$ and updates the $Q$-values for the state-action pair $(s_t,a_t)$ by,
\begin{equation*}
Q_{k+1} = (1-\alpha)Q_k(s_t,a_t)+\alpha\underbrace{\left(r_t+\gamma\max_{a'}Q_k(s_{t+1},a')\right)}_{target}.
\end{equation*}
The pseudo-code for the tabular $Q$-learning is given below,

\begin{algorithm}
\caption{Tabular Q-Learning}\label{Tabular_Q}
\begin{algorithmic}
\State Initialize a random state distribution $D_0$.
\State Initialize an array $\hat{Q}$ for estimates of the $Q$-value.
\While {change in $\hat{Q}$ over consecutive episodes is not small}
\State Reset $t=1$, sample $s_1\sim D_0$.
\For{$t=1,\cdots$ end of episode}
\State Take an action $a_t$, observe reward $r_t$, new state $s_{t+1}$.
\State $\text{target}\coloneqq \left(r_t+\gamma\max_{a'}\hat{Q}(s_{t+1},a')\right)-\hat{Q}(s_t,a_t).$
\State Update $\hat{Q}(s_t,a_t)\leftarrow\hat{Q}(s_t,a_t)+\alpha*\text{target}$
\EndFor
\EndWhile
\State \textbf{end}
\end{algorithmic}
\end{algorithm} 

\begin{remark}\textbf{Convergence.}
For the infinite horizon with finite state and section case, 
the convergence of the $Q$-learning algorithm can be proven 
if all actions and states are sampled infinitely many times 
(it does not matter how we select the actions), 
learning rates for updating the $\hat{Q}$ are small and 
do not decrease too quickly (we considered a constant learning rate in the pseudo-code above). 
The proof relies on a tool from online optimization called \textit{ the stochastic approximation method} (see \cite{RL_book2}, \cite{sutton2018reinforcement}).
\end{remark}

\begin{remark}
\textbf{Action exploration.}The $Q$-learning convergence theorem states that it does not matter how the actions are selected as long as they are sampled infinitely many times. One can choose a greedy policy per equation (\ref{greedy-policy}). However, this has the drawback of compounding errors, namely, greedy actions based on sub-optimal $Q$-estimates may result in getting stuck in local optima and failure to estimate $Q$-values with higher values.

One solution to this is to sample an \textit{$\epsilon$-greedy policy} at each step. 
In the pseudo-code above, before updating the $\hat{Q}$ estimate, one varies the policy by randomization,
\begin{equation}
\pi^{k+1}(s) = 
\left\{
\begin{array}{lr}
a^*_k\sim \argmax_{a}\hat{Q}^{\pi^k}(s,a) & \text{probability}~ 1-\epsilon\\
\text{randomly sampled $a$} & \text{probability}~ \epsilon
\end{array}
\right. ,
\end{equation}
\end{remark}
and the \textit{target} in the $\hat{Q}$-update takes the form of
\begin{equation*}
\text{target}\coloneqq \left(r_t+\gamma\E_{a\sim\pi^{k+1}(s_{t+1})}\hat{Q}(s_{t+1},a)\right)-\hat{Q}(s_t,a_t).
\end{equation*}
\subsubsection{Q-Learning with Function Approximation}

The $Q$-learning method introduced above constructs a lookup table of the size $|\mathcal{S}|\times|\mathcal{A}|$. This table is prohibitively large for almost any reasonably interesting MDP in practice (curse of dimensionality). An approach to tackle the curse of dimensionality in tabular Q-learning is to approximate the optimal Q-function with function approximators \cite{RL_book2}. For example, the function approximation $Q_\theta(s,a)$ could be parametrized as a linear regression, a regression/decision tree or a deep neural network (hence the name deep Q-network -- DQN \cite{DQN_paper}). With a given function representation $Q_\theta(s,a)$, the estimates for the $Q$-values for any state-action pair can be computed. As a result, in $Q$-learning with function approximators, instead of populating the $|\mathcal{S}|\times|\mathcal{A}|$ table, we approximate the parameters of the function approximation to the $Q$-function.

Recall that by Bellman optimality equation for the optimal Q-function \eqref{Q_Bellman_equation}, we want to have the following for any state-action pair,
\begin{equation}
Q_\theta^*(s,a) \approx E_{s'\sim P(.|s,a)}\left[r(s,a,s')+\gamma\max_{a'}Q_\theta^*(s',a')\right],
\end{equation}
To a given transition $(s,a,s')$, we assign the error,
\begin{equation}
l_\theta(s,a,s') = \left(Q_\theta (s,a)-\underbrace{[r(s,a,s')+\gamma\max_{a'}Q_\theta(s',a')]}_{target}\right)^2.
\end{equation}
The goal is to minimize the \textit{Bellman Mean Squared Error},
\begin{equation}
\ E_{s'\sim P(.|s,a)}\left[l_\theta(s,a,s')\right].
\end{equation}

In Q-learning with function approximation, upon taking an action and receiving rewards and new state, instead of updating the $Q(s,a)$ table by taking the running average, one updates the parameters in $Q_\theta(s,a)$. In training DQN, updating the weights can be done by stochastic gradient descent (SGD). A description of the algorithm is given below in (\ref{Q_learning}).

\begin{algorithm}[H]
\caption{Q-Learning with Function Approximation and SGD}\label{Q_learning}
\begin{algorithmic}
\State Initialize $Q_\theta(s,a)$ (differentiable with respect to $\theta$)
\State Get initial state $s_0$
\For{$k=0,\cdots$ until convergence}
\State Sample action $a$ and observe reward $r$ and next state $s'\sim P(.|s,a)$
\State $\text{target}\coloneqq \left(r_t+\gamma\max_{a'}Q_{\hat{\theta}}(s_{t+1},a)\right)$
\State $l_{\theta_k}(s, a, s'):=\left(Q_{\hat{\theta}} (s,a)-\text{target}\right)^2$
\State $\theta_{k+1}\leftarrow\theta_k-\alpha\nabla_{\theta_k}l_{\theta_k}(s, a, s')$
\If{$s'$ is terminal state}
\State $s\leftarrow s_0$
\Else{}
\State $s\leftarrow s'$
\EndIf
\EndFor
\end{algorithmic}
\end{algorithm} 

\subsubsection*{Tricks for Training DQNs}
\begin{itemize}
\item \textbf{Experience Replay.}
In the $Q$-learning with function approximation, the actions are generated in an $\epsilon$-greedy fashion from the estimate of the $Q$-function at the current step.
This usually leads to correlated samples. To tackle this issue, experience replay method is used. At each time step $t$, the transition $(s_t, a_t, r_t, s_{t+1})$ 
is saved into the replay memory $\mathcal{M}$. 
In deep $Q$-learning, at each training episode, a random mini-batch from $\mathcal{M}$ is sampled to train the neural network via stochastic gradient descent.

\item \textbf{Training Stability.} By examining the Bellman squared error loss, it is seen that ground truth for the loss minimization are generated on the fly by reusing the estimate for the $Q_{\theta}$ in the target term. 
This may cause instabilities in the optimization as the target may change too quickly based on changes in $\theta$ through gradient descent. 
A trick to tackle this is the \textit{lazy update} of target network: the parameters in the $Q$-function as part of the target are updated less frequently.
\end{itemize}

The pseudo-code for the full DQN training algorithm is given below \cite{DQL_mathematical},
\begin{algorithm}[H]
\caption{DQN Training, with Exploration, Experience Replay and Lazy Target Updates}\label{DQN_training}
\begin{algorithmic}
\State \textbf{Input:} MDP setup $(\mathcal{S},\mathcal{A}, P, R, \gamma, T)$, Replay buffer $\mathcal{M}$, mini-batch size $n$, exploration parameter $\epsilon$, neural network architecture $Q_{\theta}$, $T_{\text{target}}$ for target update frequency, a sequence of SGD learning rates $\{\alpha_t\}$. 
\State \textbf{Initialize:} Replay buffer array, $Q$-network weights $\theta$, target network weights $w=\theta$.
\For{$\text{episode}=1,\cdots,K$}
\State Observe initial state $s_0$,
\For{$t=0,\cdots T-1$}
\State Choose an action $a_t$ according to an $\epsilon$-greedy policy.
\State Execute $a_t$ and observe reward $r_t$ and the next state $s_{t+1}$.
\State Store the transition $(s_t, a_t, r_t, s_{t+1})$ in the replay buffer $\mathcal{M}$.
\State Experience replay: Sample a random mini-batch of transitions $\{(s_i, a_i, r_i, s'_{i}\}_{i=1}^{n}$ from $\mathcal{M}$.
\State For each $i$, compute the target $Y_i\coloneqq \left(r_i+\gamma\max_{a'}Q_{w}(s_{i+1},a)\right)$
\State $\theta\leftarrow\theta-\alpha_t*\frac{1}{n}\sum_{i=1}^n\left[Y_i-Q_{\theta}(s_i, a_i)\right]*\nabla_{\theta}Q_{\theta}(s_i,a_i)$.
\State Update the target network weights every $T_{\text{target}}$ steps: $w\leftarrow\theta$
\EndFor
\EndFor
\end{algorithmic}
\end{algorithm}

\subsection{Policy Gradient Methods}
Methods in the previous section fall under the dynamic programming approach to RL (\cite{sutton2018reinforcement}). In tabular and deep Q-learning methods, the goal is to approximate the optimal Q-function using the dynamic programming structure of the MDP. The optimal policy is then the greedy policy according to the optimal Q-function.

In the policy gradient approach, the policy is approximated directly. For an MDP defined by $(\mathcal{S},\mathcal{A},R, P, \gamma, T)$, let $\pi_{\theta}:\mathcal{S}\rightarrow\mathcal{P}(\mathcal{A})$ denote a stochastic policy parameterized by $\theta$ (this also includes the deterministic policies as a special case). This parametrization can be a deep neural network for instance, defining a distribution over the set of actions for each state.

Let $V^{\pi_{\theta}}$ denote the value function corresponding to the long term discounted reward collected by the agent following policy $\pi_\theta$ starting from the state $s$. Therefore, the goal is to find the optimal policy through solving the following optimization problem.

\begin{equation}
\max_{\theta}V^{\pi_\theta}(s) = \max_{\theta}\E\left[\sum_{t=0}^{T}\gamma^t r_t \middle| s_0=s,\pi_\theta\right]
\end{equation}

The following classical result (\cite{REINFORCE_paper}) calculates the gradient of the objective above with respect to $\theta$. Consider a trajectory $\tau=\{s_1,a_1,\cdots,a_{T-1},s_T\}$ resulting from rolling out the policy $\pi_\theta$. Define the accumulated rewards along the path as $R(\tau)$, also the probability of the trajectory as
\begin{equation*}
D(\tau)=\prod_{i=1}^{T-1}\pi(a_i|s_i)P(s_{i+1}|s_i,a_i).
\end{equation*}
\begin{theorem}\label{PG_theorem}\textbf{(Policy Gradients.)}
Consider the finite horizon roll-out of the policy $\pi_\theta$ from the initial state $s$. Then it holds that,
\begin{equation*}
\nabla_{\theta}V^{\pi_{\theta}}(s)=\E_{\tau}\left[R(\tau)\nabla_{\theta}\log(D^{\pi_\theta}(\tau))\middle|s_0=s\right] =\E_{\tau}\left[R(\tau)\sum_{t=0}^{T-1}\nabla_{\theta}\log(\pi_\theta(a_t|s_t))\middle|s_0=s\right].
\end{equation*}
\end{theorem}
\begin{proof}See \cite{sutton2018reinforcement} or \cite{REINFORCE_paper} for details.
\end{proof}

\begin{remark}
The above representation of the gradient of the objective 
is useful because it allows for calculating the gradient by rolling out the policy and estimating an expected value. 
Note that this does not need any separate and detailed knowledge of the state transition density or the reward model, 
one can use the realized states and reward values along the observed trajectories. 
Also note that for a neural network parametrization of $\pi_\theta$, 
the term $\nabla_{\theta}\log(\pi_\theta(a_t|s_t))$ can be calculated efficiently via backpropagation.
\end{remark}

\begin{remark}
The characterization of the gradient in Theorem \ref{PG_theorem} can be stated in a more general form. Let $X$ be a random variable with known p.d.f. $p_\theta(X)$. Let $J(\theta)$ be some cost function defined as
\begin{equation*}
J(\theta) = \E_{\theta}[f(X)]=\int_x f(x)p_\theta(x)dx.
\end{equation*}
for some arbitrary function $f(X)$ that does not depend on $\theta$. Then one has,
\begin{equation*}
\nabla_\theta J(\theta) =\int_x f(x)\nabla_\theta p_\theta(x)dx =\int_x f(x)p_\theta(x)\nabla_\theta \log(p_\theta(x))dx= E[f(X)\nabla_\theta \log(p_\theta(x)].
\end{equation*}
\end{remark}

Policy gradient formula gives an unbiased estimate of $\nabla_\theta V^{\pi_\theta}$,
\begin{equation}\label{PG_estimator}
\boldsymbol{\hat{g}}=\frac{1}{m}\sum_{i=1}^m\hat{R}(\tau_i)\sum_{t=0}^{T-1}\nabla_\theta\log(\pi_\theta(a^i_t|s^i_t)).
\end{equation}
This leads us to the vanilla version of the REINFORCE Algorithm (\cite{REINFORCE_paper}) for policy gradients.

\begin{algorithm}[H]
\caption{REINFORCE Algorithm}\label{REINFORCE}
\begin{algorithmic}
\State \textbf{Input} Policy parametrization $\pi_\theta(a|s)$
\State \textbf{Initialize} $\theta\in\mathbb{R}^m$
\For{$k=0,\cdots$ until convergence}
\State Generate episode $\tau=\{s_1,a_1,r_1\cdots,s_{T-1},a_{T-1},r_{T-1},s_T\}$
\State Calculate $\hat{R}(\tau) = r_1+\gamma r_2+...+\gamma^{T-1}r_{T-1}$
\For{each step $t=0,\cdots,T-1$}
\State $\theta\leftarrow \theta+\rho\hat{R}(\tau) \nabla_\theta\log(\pi_\theta(a_t|s_t))$ 
\EndFor
\EndFor
\end{algorithmic}
\end{algorithm} 

While REINFORCE algorithm works in theory, in practice, it suffers from high variance. 
The following property of the gradients enables us to lower the variance of the estimator without introducing bias.

\begin{align*}
\E\left[ \nabla_\theta\log(\pi_\theta(a_i|s_i))\right]&=\int_{a_i}\pi(a_i|s_i) \nabla_\theta\log(\pi_\theta(a_i|s_i))da_i\\
&=\int_{a_i}\nabla_\theta\pi_\theta(a_i|s_i)da_i=\nabla_\theta\int_{a_i}\pi_\theta(a_i|s_i)da_i\\
&=\nabla_\theta1=0.
\end{align*}

Using the above property, one can add a "baseline" $b$ in the empirical policy gradient without introducing bias in $\nabla_\theta(V^{\pi_\theta})$:
\begin{equation}
\boldsymbol{\hat{g}_b}=\frac{1}{m}\sum_{i=1}^m\left((\hat{R}(\tau_i)-b)\sum_{t=0}^{T-1}\nabla_\theta\log(\pi_\theta(a^i_t|s^i_t)\right).
\end{equation}

More generally, the baseline can be state-dependent,
\begin{equation}
\boldsymbol{\hat{g}_b}=\frac{1}{m}\sum_{i=1}^m\sum_{t=0}^{T-1}\left((\hat{R}(\tau_i)-b_t(s_i))\nabla_\theta\log(\pi_\theta(a^i_t|s^i_t)\right).
\end{equation}

We end this section by observing that one can further exploit the temporal structure of an MDP 
to further reduce the variance of the policy gradient estimator \eqref{PG_estimator}.
The idea is to split the path-wise sum of rewards into parts prior to step $t$ and after step $t$.
\begin{align*} 
\boldsymbol{\hat{g}}&= \frac{1}{m}\sum_{i=1}^{m}\left(\sum_{t=1}^{T-1}\nabla_\theta\log(\pi_\theta(a_t^i|s_t^i))\right)\left(\sum_{t=0}^{T-1}(r(s_t^i,a_t^i)-b(s_t^i))\right) \\ 
&=\frac{1}{m}\sum_{i=1}^m\left(\sum_{t=1}^{T-1}\nabla_\theta\log(\pi_\theta(a_t^i|s_t^i))\left[\left(\sum_{k=0}^{t-1}r(s_k^i,a_k^i)+\sum_{k=t}^{T-1}r(s_t^i,a_t^i)\right)-b(s_t^i)\right] \right)
\end{align*}

Removing terms in the reward summation that do not depend on the current action will help reduce variance. 
Therefore, denoting $G_t=\sum_{k=t}^{T-1}r(s_t^i,a_t^i)$, we have the following form for the policy gradient estimator,

\begin{equation}\label{PG2}
\boldsymbol{\hat{g}}=\frac{1}{m}\sum_{i=1}^{m}\sum_{t=1}^{T-1}\nabla_\theta\log(\pi_\theta(a_t^i|s_t^i))\left( G_t-b(s_t^i) \right).
\end{equation}

\subsubsection{Deep Deterministic Policy Gradients}
In this section, we briefly review the Deep Deterministic Policy Gradient (DDPG) method introduced in \cite{DDPG_paper}. 
DDPG can be thought as the extension of Deep Q-learning to continuous action spaces. 
Recall that in Deep Q-learning, the goal is to find the optimal action-value function $Q^*(s,a)$ 
and then finding the optimal policy as the greedy policy with respect to $Q^*$, 
needing to solving a maximization problem $\argmax_a Q^*(s,a)$. 

In the case of finite discrete actions, solving the maximization problem is possible. 
However, when the action space is continuous, this needs exhaustive search over the continuous space which is typically not feasible. 
In addition, implementing common optimization algorithms would make calculating $\argmax_a Q^*(s,a)$ a computationally expensive subroutine.
DDPG algorithms tackle this problem by using a target policy network to propose an action 
which tries to approximate $\argmax Q_{\phi_{\text{targ}}}$. 

On the policy gradient side, policy learning in DDPG aims to learn a deterministic policy $\mu_{\theta}(s)$ producing actions maximizing $Q_{\phi}(s,a)$. DDPG assumes that the Q-function is differentiable with respect to action, therefore, learning is done by performing gradient ascent (with respect to policy parameters only and treating Q-function parameters as constants) to solve,
\begin{equation}
\max_{\theta} \underset{s \sim {\mathcal D}}{{\mathrm E}}\left[ Q_{\phi}(s, \mu_{\theta}(s)) \right].
\end{equation}
Note that the above formulation allows the policy to be both continuous and deterministic.
Below, we present the pseudo-code summary of the DDPG algorithm \footnote{Adapted from \url{https://people.eecs.berkeley.edu/~pabbeel/cs287-fa19/}.},

\begin{algorithm}
\caption{Deep Deterministic Policy Gradient Algorithm}\label{Q-A2C}
\begin{algorithmic}
\State Initialize policy parameters $\theta$, Q-function parameters $\phi$, empty replay buffer $\mathcal{M}$
\State Set target parameters equal to main parameters $\theta_{\text{targ}}\leftarrow\theta$, $\phi_{\text{targ}}\leftarrow\phi$
\State Initialize replay buffer $\mathcal{M}$
\For{$\text{episode}=1,\cdots,K$}
\State Initialize a random process $\mathcal{N}$ for exploration
\State Observe initial state $s_1$ 
\For{$t=1,\cdots,T(=\text{time horizon})$}
\State Select action $a_t=\mu_{\theta}(s)+\mathcal{N}_t$
\State Execute action, receive reward $r_t$ and state transition $s_{t+1}$
\State Append the transition to the replay buffer $\mathcal{M}$
\State Sample a mini-batch of transitions $B= \{(s_i, a_i, r_i,s_{i+1})\}$ from the replay buffer $\mathcal{M}$
\State For each transition tuple, set, $\text{target} = r+\gamma Q_{\phi_{\text{targ}}}\left(s',\mu_{\theta_{\text{targ}}}(s')\right)$, where $s'=s_{i+1}$
\State \textbf{Q-update:} Update $Q$-function network weights by one-step gradient descent,
\begin{equation*}
\nabla_{\phi}\frac{1}{|B|}\sum_{(s, a, r,s')\in B}\left(Q_{\phi}(s,a)-\text{target}\right)^2
\end{equation*}
\State Get updated weights $\phi^{\text{updated}}$
\State \textbf{Policy update:} Update policy network weights by one-step gradient ascent, 
\begin{equation*}
\nabla_{\theta}\frac{1}{|B|}\sum_{s\in B}Q_{\phi}(s,\mu_{\theta}(s))
\end{equation*}
\State Get updated weights $\theta^{\text{updated}}$
\State Update target networks
\begin{equation*}
\phi_{\text{targ}}\leftarrow\rho\phi_{\text{targ}}+(1-\rho)\phi^{\text{updated}}
\end{equation*}
\begin{equation*}
\theta_{\text{targ}}\leftarrow\rho\theta_{\text{targ}}+(1-\rho)\theta^{\text{updated}}
\end{equation*}
\EndFor
\EndFor
\end{algorithmic}
\end{algorithm}

\newpage
\section{Deep Trajectory-Based Stochastic Optimal Control}

Deep Trajectory-Based Stochastic Optimal Control (DTSOC), proposed in \cite{han2016deepsc} 
(also see \cite{Deep_stochastic} for an exposition), is a method for solving stochastic control problems 
through formulating the control problem as optimizing over a computational graph, with the sought controls represented as (deep) neural networks. 
The approximation power of the deep neural networks can mitigate the curse of dimensionality for solving dynamic programming problems.

We briefly review the setup of the method here. Consider a stochastic control problem given by the following underlying stochastic dynamics,
\begin{equation}
s_{t+1} = f(s_t,a_t,\xi_t)
\end{equation}
where, $s_t$ is the state, $a_t$ is the control (agent's action) and $\xi_t$ is a stochastic disturbance impacting the transition 
between times $t$ and $t+1$.  

In the models for the hedging instruments derived from (discretized) SDEs, the $\xi_t$ will be the Brownian increments in the (discretized) SDEs for time $t$ 
respective the time step from $t$ to $t+1$, $dW_t$ respective $\Delta W_i$. The price of the hedging instrument as well as the amount of them held would be part of the 
state and the action would either directly give the new amount to be held or an increment or rate that would allow the new amount to be computed. 

We assume that the actions are given as deterministic or stochastic feedback controls 
$a_t = \pi_t(s_t|\theta_t)$
or
$a_t \sim \pi_t(\cdot|s_t,\theta_t)$.
One can extend the state $s_t$ with path-dependent extra state that can be computed from current and previous state, action, and disturbances; 
and also with particular precomputed features that might lead to more efficient training of agents or more efficient agents which 
also extends the set of controls that can be written as feedback controls. 

The actions can be constrained to come from a set of admissible actions: 
\begin{equation*}
a_t\in\mathcal{A}_t=\left\{a_t: 
g(s_t,a_t)=0, h(s_t,a_t)\geq0\right\},
\end{equation*} 
where $h(s_t,a_t)$ and $g(s_t,a_t)$ are inequality and equality constraints. We assume that these constraints are already taken into account in the feedback controls such 
that $\pi_t$ will be an admissible action or give a distribution over admissible actions. 

We assume a stepwise cost (or negative reward) function is specified $c_t=c_t(s_t,a_t,s_{t+1})$ 
and also a final cost $c_T(s_T)$.

Given a deterministic or probabilistic policy in feedback form that gives admissible actions, we generate 
an episode 
\[
s_0, a_0, s_1, a_2, \cdots, a_{T-1}, s_T
\]
and obtain a cumulative cost
\begin{equation}
C= \sum_{t=0}^{T-1} c_t(s_t,a_t,s_{t+1}) + c_T(s_T).  
\end{equation}

The stochastic optimal control problem now minimizes the expected cumulative cost, conditional on starting state $s_0$. If $s_0$ is not 
fixed, this will be a function of $s_0$. 
We thus try to minimize $\E[C|s_0=s]$ or $\E[C]$ varying the policies $\pi_t$. 

With some given functional form (such as DNN) with an appropriate parametrization (for example,
determine weights and biases while activation functions are fixed for complete feedforward DNN) as deterministic policy, we obtain
\begin{equation}
 \E\left[C\right] = \E\left[ \sum_{t=0}^{T-1} c_t(s_t,\pi_t(s_t|\theta_t),s_{t+1})+ c_T(s_T) \right] =: {\cal L}\left(\left\{ \theta_t\right\}_{t=0}^{T-1} \right) =: {\cal L}(\Theta) \label{DNN_loss} 
\end{equation}

One now jointly optimizes over all policies $\{\pi_t\}^{T-1}_{t=0}$ respective over all parameters of such $\left\{ \theta_t\right\}_{t=0}^{T-1}$ to optimize the above 
cumulative costs (if $s_0$ is not fixed, this will also be a function of $s_0$ and we would need to take an appropriate expectation over $s_0$ or keep $s_0$ as a parameter).

The controls at each time step could be stacked into a computational graph with a loss function given in \eqref{DNN_loss}. 
For each roll out of the control problem, this computational graph takes the sequence of disturbances $\{\xi\}_{t=0}^{T-1}$ as input and gives the accumulated cost inside the expectation in \eqref{DNN_loss} as output. As demonstrated in Figure \ref{DNN_control_arch}, the computational graph has the following features:

\begin{itemize}
\item The deterministic policy at time step $t$ is represented by some network with appropriate architecture, here for example a complete feed-forward neural network, $s_t\rightarrow a_t$ with parameters $\theta_t$
\item The transition of the system to a new state, $(s_t, a_t)\rightarrow s_{t+1}$ based on the system dynamics is encoded in the inter-block connections between $a_t$ and $s_{t+1}$.
\item Defining the cumulative cost up to time $t$ as,
\begin{equation*}
C_t = \sum_{\tau=0}^t c_\tau(s_\tau,a_\tau,s_{\tau+1}), 
\end{equation*}
the horizontal connections on top of the network, $(s_t,a_t,C_t)\rightarrow C_{t+1}$ sums up the immediate costs and gives the total accumulated cost at the end of the episode (when $t=T$).
\item There are also connections from $s_t$ and $s_{t+1}$ to $C_t$ but they are not shown in the figure to avoid clutter. 
\end{itemize}

Note that based on a discretization $\{0=t_0<t_1,\cdots,t_p=T\}$ of the time horizon, the computational graph will have $p$ embedded DNN. 
If these embedded DNN do not share parameters and have $NP$ trainable parameters each, the entire computational graph has then $p \times NP$ trainable 
parameters.  
\begin{figure}[!htbp]\label{DNN_control_arch}
\begin{center}
\fbox{\includegraphics[height =6.5cm,width=14cm]{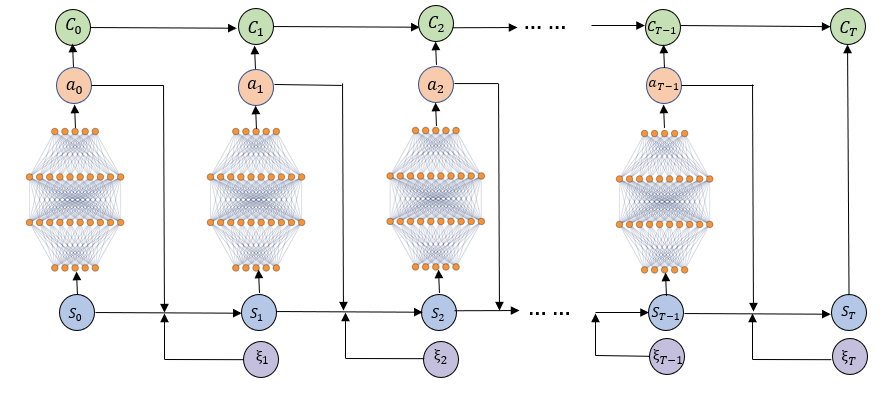}}
\caption{Neural computational graph for DTSOC.}
\end{center}
\end{figure}

If desired and appropriate for the application, one can introduce an additional future discount factor $\gamma$ which balances the importance of optimizing immediate versus future costs by setting 
$\tilde{c}_t= \gamma^t c_t$.  

The pseudo-code for training the neural network computational graph is described in Algorithm \ref{DSOC_training} below. 

\begin{algorithm}[H]
\caption{Training Procedure for DTSOC}\label{DSOC_training}
\begin{algorithmic}
\State \textbf{Data:} The model and parameters for the underlying dynamics 
\State \textbf{initialize:} Network weights $\{\theta_m, m=0,\cdots,T-1\}$
\While {epoch$\leq EPOCH$}
\State $nbatch=0$
\State $batchloss=0$
\While {$nbatch\leq \text{batch size}$}
\State $m=0$
\State $C_{-1}=0$
\While {$m<T$}
\State $a_m = F^{\theta_{m}}(s_{m})$\
\State $s_{m+1} = f (s_{m}, a_{m},\xi_m)$ with sampled noise $\xi_m$. 
\State $C_m=C_{m-1} + c_m(s_{m},a_{m},s_{m+1})$
\State $\text{m++}$
\EndWhile
\State $C_T=C_{T-1} + c_T(s_T)$
\State $batchloss+=C_T$
\EndWhile
\State Calculate loss for batch $$\text{Loss}=batchloss/(\text{batch size})$$
\State Calculate gradient of $\text{Loss}$ with respect to $\theta$
\State Back propagate updates for $\{\theta_m,m=0,\cdots,T-1\}$ 
\State $\text{epoch++}$
\EndWhile
\State \textbf{return} optimized weights $\{\theta^*_m,m=0,\cdots,T-1\}$ 
\end{algorithmic}
\end{algorithm}

\subsection{Relationship to FBSDE Formulation of Stochastic Control}

In general, one can consider a stochastic control problem in which some functional defined by running and final costs 
(which depend on the evolution of some controlled forward SDE and on the control) is optimized over that control. 
This leads  to coupled \textit{forward backward stochastic differential equations} (FBSDE) and non-linear PDEs 
(See \cite{perkowski2011backward} or \cite{pham2009continuous} for an introductory treatment). In our setting, the control 
that the agent tries to optimize does not impact the forward SDEs describing the evolution of the prices of the 
instruments, it only impacts the trading strategy, leading to a controlled backward SDE only.  

With $X_t$ being the factors and prices for the (hedging) instruments and $Y_t$ being the value of the 
hedging 
strategy, we have the system (see \cite{hientzsch2019introduction})\footnote{
In this subsection, we use notation from the FBSDE literature as adapted to the pricing and hedging domain and do not follow the generic 
notation for RL or trajectory-based approaches. The state $s_t$ in RL or trajectory-based approaches 
would contain $X_t$, $Y_t$, $\Pi_t$, $J_t$, and whatever is needed to compute terms and costs (or equivalent information), the action/control
would be some parametrization of $\Pi_t$, the stochastic disturbance $\xi_t$ would be the $dW_t$ or $\Delta W^i$.
}, 
\begin{equation}
dX_{t} = \mu\left( t,X_{t} \right)dt + \sigma\left(t,X_{t}\right)dW_{t}
\end{equation}
\begin{equation}
 dY_{t} = - f\left( t,X_{t},Y_{t},\Pi_{t} \right)dt\  + \Pi_{t}^{T}
\sigma\left(t,X_{t}\right) dW_{t}
\end{equation}
where $\Pi_{t}$ plays the role of a control or strategy and the functional to be optimized (typically, minimized) is 
\begin{equation}
J^F(\Pi_t,\Pi^{\mathsf{final}},\ldots) = E\left(\int_0^T
{\mathsf{rc}}(s,{X}_s,Y_s,\Pi_s) ds + {\mathsf{fc}}({X}_T,Y_T,\Pi^{\mathsf{final}})\right).
\end{equation}
For the example of an European option, $Y_T$ has to replicate the appropriate payoff $g(X_T)$. 

One can define 
\[ J_t =  \int_0^t {\mathsf{rc}}(s,{X}_s,Y_s,\Pi_s) ds \]
or 
\[ dJ_t = {\mathsf{rc}}(t,{X}_t,Y_t,\Pi_t) dt \]
and add it to the stochastic system, looking for a minimum of 
\[ \E \left( J_T +  {\mathsf{fc}}({X}_T,Y_T,\Pi^{\mathsf{final}})\right).\]

One can derive FBSDE characterizing the optimal controls (both primal and dual/adjoint) as well as
PDEs characterizing them, but we will here concentrate on approaches that try to directly optimize 
over the given system for $X_t$, $Y_t$, $\Pi_t$, and $J_t$. 

Upon time-discretization, one obtains stochastic control problems defined on (controlled) FBS$\Delta$E ($\Delta$ standing for "difference") where now the 
running cost can depend on the forward and backward components and the control at both the beginning and end of each time-period. 

Applying a simple Euler-Maruyama discretization for both $X_t$ and $Y_t$, we
obtain 
\begin{equation}
X_{t_{i+1}} = X_{t_i} + \mu(t_i,X_{t_i}) \Delta t_i + \sigma^T(t_i,X_{t_i})
\Delta W^i
\end{equation}
\begin{equation}
Y_{t_{i+1}} = Y_{t_i} - f\left( t_i,X_{t_i},Y_{t_i},\Pi_{t_i} \right) \Delta 
t_i + \Pi^T_{t_i} \sigma^T(t_i,X_{t_i}) \Delta W^i
\end{equation}
This can be used to time-step both $X_t$ and $Y_t$ forward. 

Now, $Y_T$ in general can no longer perfectly replicate $g(X_T)$ but still needs to be constrained to replicate the instrument 
in an appropriate sense, such as by minimizing the squared differences $E(|Y_T-g(X_T)|^2)$. 

Similarly, the running costs need to be accumulated
\begin{equation}
J_{t_{i+1}} = J_{t_i} + {\mathsf{rc}}(t_i,{X}_{t_i},Y_{t_i},\Pi_{t_i}) \Delta t_i
\end{equation}
and the stochastic optimal control problem will try to minimize
\[ \E \left( J_T +  {\mathsf{fc}}({X}_T,Y_T,\Pi^{\mathsf{final}})\right).\]

A time-discrete setting allows one to incorporate more general transaction costs for $Y_t$ \cite[Section 7.2]{hientzsch2019introduction} 
by more complicated generators $f$
\begin{equation}
Y_{t_{i+1}} = Y_{t_i} - f_{\Delta t} \left( t_i,\Delta t_i,
X_{t_i},X_{t_{i+1}},Y_{t_i},Y_{t_{i+1}},\Pi_{t_i},\Pi_{t_{i+1}} \right) 
+ \Pi^T_{t_i} \sigma^T(t_i,X_{t_i}) \Delta W^i
\end{equation} 
and also more complicated running costs 
\begin{equation}
J_{t_{i+1}} = J_{t_i} + {\mathsf{rc}}\left(t_i,\Delta t_i,{X}_{t_i},X_{t_{i+1}},Y_{t_i},Y_{t_{i+1}},\Pi_{t_i},\Pi_{t_{i+1}}\right) \Delta t_i,
\end{equation}
which could include running costs that depend on the profit and loss of some strategy across the corresponding time interval. 

Since with appropriate generator, final values, and settings, the backward component $Y$ corresponds to a replication or risk-management 
price of the instrument under appropriate models, one could use the solution $Y$ of the  FBS$\Delta$E and/or P$\Delta$E 
in appropriate form as the book-keeping and/or consistent model in
hedging or risk management setups that require such models, as in the stepwise mean-variance hedging setup discussed here. 
For such setups, it might be advantageous  to introduce several $Y$ processes, one reflecting a book-keeping or consistent 
model according to some set of assumptions 
(and thus providing such a book-keeping or model price for use after applying standard FBSDE techniques) while one or several
others reflect the performance of other strategies that will be measured against such models. 

In \cite{weinan2017deep,hientzsch2019introduction,hientzsch2021intro,ganesan2019pricingbarriers,ganesan2022pricingbarriers,liang2019deep,liang2021deep}, 
path-wise deepBSDE methods for such problems are discussed, at least applied to pricing and risk management where 
there is only a final cost (or a cost at the earlier of reaching a barrier or maturity) - as in a quadratic hedging setup. 

DeepBSDE methods represent the strategy $\Pi_t$ as a DNN depending on appropriate state $X_t$ (or features computable from such state). 
Path-wise forward deepBSDE methods generate trajectories of $X$ and $Y$ starting from initial values $X_0$ and $Y_0$ 
according to the current strategy. They then use stochastic gradient  descent type approaches such as ADAM to improve the strategy
until an approximate optimum is reached. 
If the initial wealth $Y_0$ is not given, it will be determined by the optimization as well. 
If the starting value $X_0$ of the risk factor vector is fixed, $Y_0$ would be a single value, 
otherwise it would be a function of $X_0$. The optimization problem would typically represent
this function as a DNN. 
Derived so far for certain kinds of final costs or where one attempts to replicate the final payoff as well as possible, path-wise backward deepBSDE 
methods make the same assumptions, but on each generated forward trajectory of $X$, they start with an appropriate final value of $Y_T$ ($Y_T=g(X_T)$ for
the final payoff case), compute a corresponding trajectory of $Y_t$ by stepping backward in time, and try to minimize the range of $Y_0$. 

Quadratic hedging for European options has been considered with forward and backward path-wise methods in \cite{liang2019deep,liang2021deep}
for linear pricing and in \cite{yuhientzsch2019backward} for nonlinear pricing (differential rates) while forward path-wise methods were
introduced earlier by  \cite{weinan2017deep}. The barrier option case is treated with forward methods in \cite{ganesan2019pricingbarriers,ganesan2022pricingbarriers}. 
We are not aware of any path-wise deepBSDE (or other FBSDE-related) method being used for stepwise mean-variance hedging problems 
and might treat this setting in future work. 

One difference between the setup discussed in this section and 
other approaches considered in the paper is that here the backward SDE or  S$\Delta$E 
is written in such a way that it uses parts of the forward model (i.e. $ \sigma^T(t_i,X_{t_i}) \Delta W^i$)
and might not as written satisfy self-financing exactly but only up to discretization accuracy. In this way, this is a (more) model-based approach. 
However, one can rewrite the BSDE for $Y$ so that it only uses the stochastic increment of $X$ (i.e., written in $dX$ rather than using model details about $X$) 
and one can rewrite the BS$\Delta$E so that it only requires observations of $X$ at trading times $t_i$ and perfectly preserves self-financing 
similarly to what we wrote in earlier subsections, obtaining methods that will be more similar to the ones discussed there.  

\section{Experimental Setup and Model Specification}
Following \cite{Kolm_Ritter}, we consider the example of a European call option with strike price $K$ and maturity $T$ on a non-dividend-paying stock. 
The strike price and option maturity are considered as fixed parameters. 
It is assumed that the risk-free rate is zero and that the option position is held until maturity. 
Rebalancing of the hedging portfolio is allowed at fixed (often regular) times 
and trades are subject to transaction costs. 
The trained hedging agent is expected to learn to hedge an option with this specific set of parameters. 
It is possible to train parametric agents that can hedge a parametrized set of options 
(such as calls with various strikes), but we will not do so here. 
We use the Black-Scholes model for the simulation environment 
where the stock dynamics is given by a geometric Brownian motion 
\begin{equation}
dS_{t} = \mu S_t \; dt + \sigma S_t \; dW_t \label{SDE:Stock},
\end{equation}
and trading (buying and selling) of stock incurs a transaction cost which is assumed to have the functional form (see \cite{Kolm_Ritter}),
\begin{equation}
\text{cost}(S_t, \delta H_t) = \alpha S_t(|\delta H_t|+\beta\times \delta H_t^2),
\end{equation}
where $\delta H_t$ is the change in the stock position.

The default parameters for the stock, the option, and the market friction 
are as given in Table \ref{defaultparameters}.
\begin{table}[htb]
\begin{center}
\begin{tabular}{| l || l | } \hline
Parameter & Value \\ \hline
$\mu$ & 5 \% (rate of return) \\ \hline
$\sigma$ & 20 \% (volatility) \\ \hline
$ir$ & 0.0 \% (interest rate) \\ \hline
$S_0$ & 100 \\ \hline
$K$ & 100 \\ \hline
$T$ & 30 (option maturity- days) \\ \hline
$\alpha$ & 0.01 (friction parameter)\\ \hline
$\beta$&0.01 (quadratic cost factor)\\ \hline
\end{tabular}
\caption{Parameters of the stock, the option, and the market friction.\label{defaultparameters}}
\label{Table1}
\end{center}
\end{table}

\subsection{ Feature Engineering and Model Architecture}\label{model_architecture}
The following variables were chosen to represent the state space of the problem, given as inputs to the RL and DTSOC agent at each time step,
\begin{itemize}
\item Time $t$, $0\leq t\leq T$,
\item Stock price at time $t$,
\item Option price $C_t$,
\item Option Delta\footnote{Computed or modeled with an appropriate model.} $\Delta_t$,
\item Current stock holding\footnote{Impacted by agent's action in previous step(s).}  $H_t$.
\end{itemize}

The input to the agent should give the agent all of the information needed for optimal decision making. 
This leads to feature engineering to identify necessary or helpful features to add to the state. 
In \cite{Kolm_Ritter}, a minimal set of variables including the current price of stock, either time to maturity or time, and current stock position is considered. 
In section \ref{test_feature_delta}, we test and discuss whether it is necessary and/or reasonable to include option Delta as a feature in the state. 

As for the neural network formulation of deep-MVH, one can either parametrize the hedge strategy (the number of stock to hold at time $t$) 
or parametrize the rebalancing rate. The two parametrizations are related through the equation, 
\begin{equation}\label{rate_to_ratio}
H_{t_{m+1}}^{\theta_{m+1}}=H_{t_{m}}^{\theta_{m}}+ \dot{H}_{t_m}^{\theta_{m}}\Delta_t.
\end{equation}
We choose to parametrize the rebalancing rate of the stock holding by a feedforward neural network with three hidden layers with 10, 15 and 10 hidden units. 
This is similar to how other implementations in the literature have used relatively shallow deep networks with three layers to parametrize 
each individual control (see \cite{han2016deepsc}, \cite{shi2021deep}).
RELU was used as the nonlinear activation function in all hidden layers. 
The actor component of the RL-DDPG agent is a feedforward neural network with 3 layers with 12 neurons in each layer. 
The critic component consists of a feedforward neural network of 3 layers with 24 neurons in each layer. 

\subsubsection{Should Option Delta be Included as a State Feature?}\label{test_feature_delta}
Since the option Delta is provided as a state space feature to both deep-MVH and RL-DDPG models, 
it is reasonable to ask what the model learns when the transaction cost is zero 
- since the optimal strategy for the continuous trading limit case (option Delta) is already given to the model as part of the input. 
To make the learning more difficult for the agents, we removed the option Delta from the state space 
and trained both RL-DDPG and the deep MVH models to hedge with no transaction cost.

Figure \ref{benchmark_rl_deep_MVH_wo_delta_feature} shows that, even without access to the option Delta, 
deep-MVH performs as well as the Delta hedge if not better. 
The RL-DDPG strategy, however, seems to underperform relative to the deep-MVH as the 
mean of the total hedging cost has shifted to the right. 
This shows that we may drop the option Delta feature at least for the deep-MVH model, 
however, we decided to keep it to have identical state space specifications for both algorithms.

\begin{figure}[h]
\begin{center}
\fbox{\includegraphics[width=10cm,height=8cm]{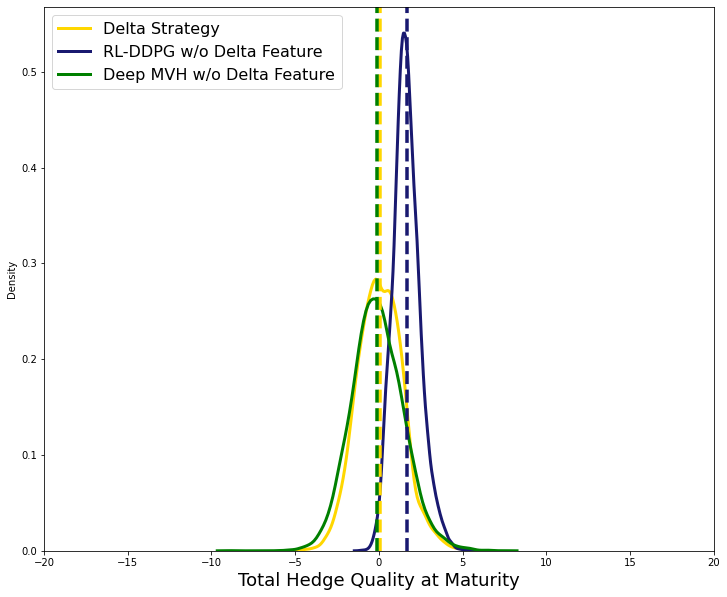}}
\caption{Histogram of total hedging cost per episode under zero transaction cost; Delta Hedge vs. deep-MVH vs. RL-DDPG without the option Delta feature as input.  The $x$-axis is the total hedging cost during the life of the option. Positive values denote loss and negative values denote profit. Deep-MVH performs as well as the Delta hedge or better, while RL-DDPG underperforms. }\label{benchmark_rl_deep_MVH_wo_delta_feature}
\end{center}
\vspace{-5mm}
\end{figure}

The box plots in Figure \ref{BP_Delta_vs_deep_MVH_no_Delta} below compare the statistics of decisions made by the deep-MVH model 
and the Delta hedge strategy in the intermediate steps. 
It is observed that if the option Delta is not provided to the deep-MVH model, 
the average decisions in terms of number of stocks to hold are reasonably close to those of Delta hedge strategy. 
It is also seen that the differences between the strategies decrease in the second half of the life of the option.

\begin{figure}[h]
\begin{center}
\fbox{\includegraphics[width=12cm,height=8cm,keepaspectratio]{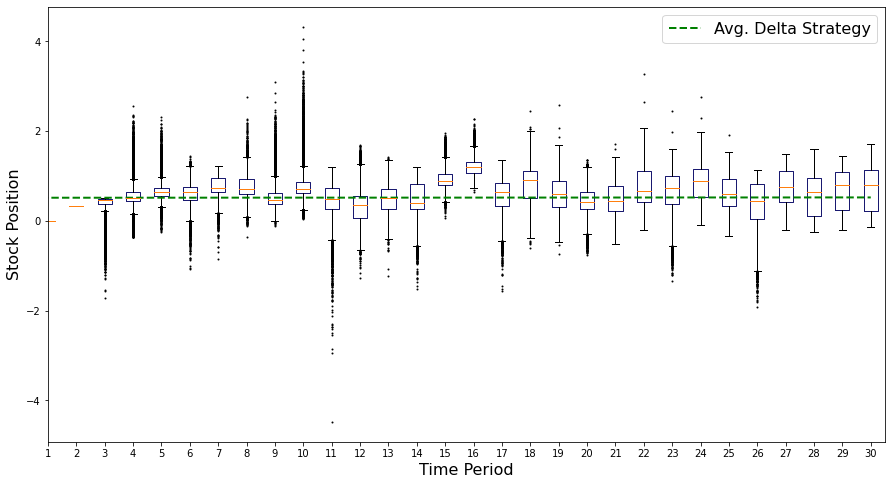}}
\caption{The box plot of decisions made by the deep-MVH model (without option Delta feature). The transaction cost is zero. The average of deep-MVH decisions (orange lines) are relatively close to average Delta. It is also observed that, in the second half of time-to-maturity, differences between the strategies decrease. }\label{BP_Delta_vs_deep_MVH_no_Delta}
\end{center}
\end{figure}

We use SHAP values (see Section \ref{SHAP} for a brief review) to gain insight how much the agents rely on each state feature 
in the absence of option Delta. 
Figure \ref{deep_MVH_explain_heatmap} shows the SHAP variable importance values for each of the 30 intermediate policies in deep MVH.

\begin{figure}[h]
\begin{center}
\fbox{\includegraphics[width=12cm,height=8cm,keepaspectratio]{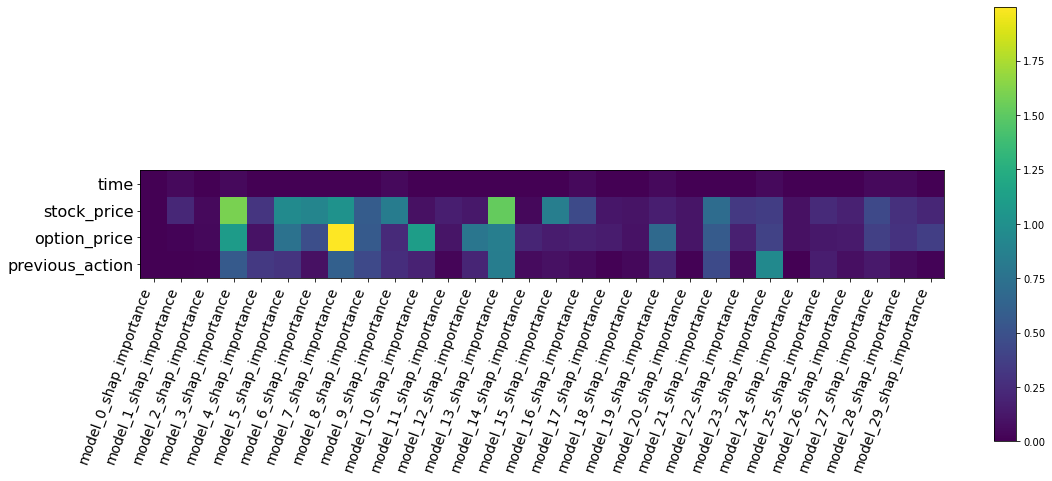}}
\caption{Heat map of SHAP values for the 30 policies in deep MVH computation graph (one per time step, no Delta feature). }
\label{deep_MVH_explain_heatmap}
\end{center}
\end{figure}

To summarize, we take average of the SHAP variable importance for each variable along all the time steps
and show the results for RL-DDPG and deep-MVH agents in Figure \ref{SHAP_RL_deep_MVH_no_delta}.

\begin{figure}[h]
\centering
\fbox{
\hfill
\subfigure{\includegraphics[height = 3.5 cm, width=6.8cm]{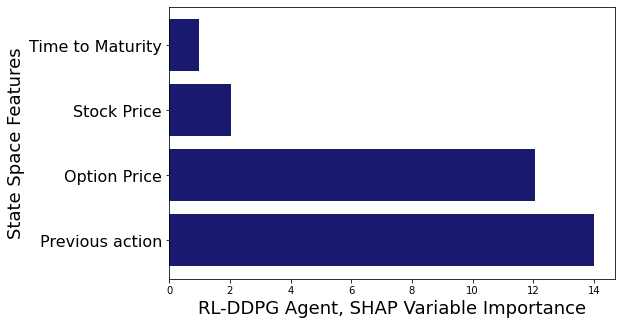}}
\hfill
\subfigure{\includegraphics[height = 3.5 cm,width=6.8cm]{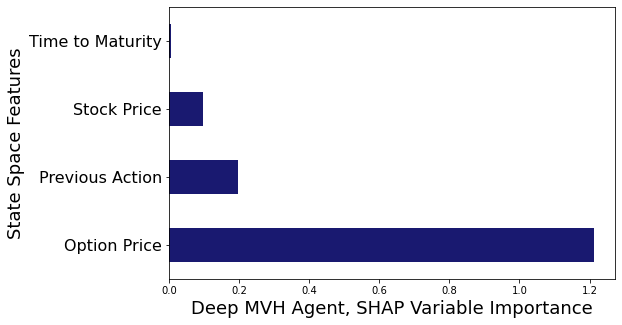}}
\hfill
}
\caption{Global SHAP variable importance for the RL-DDPG and deep MVH without option Delta feature (zero transaction cost).}\label{SHAP_RL_deep_MVH_no_delta}
\end{figure}

We repeated the same test as above but this time with non-zero transaction cost ($\alpha=0.01$). 
Figure \ref{benchmark_rl_deep_MVH_wo_delta_feature_nonzero_TC} shows that without option Delta as input, RL-DDPG agent performs worse than deep-MVH also when 
transaction costs are present, but still better than the Delta hedge strategy.  

\begin{figure}[h]
\begin{center}
\fbox{\includegraphics[width=10cm,height=8cm]{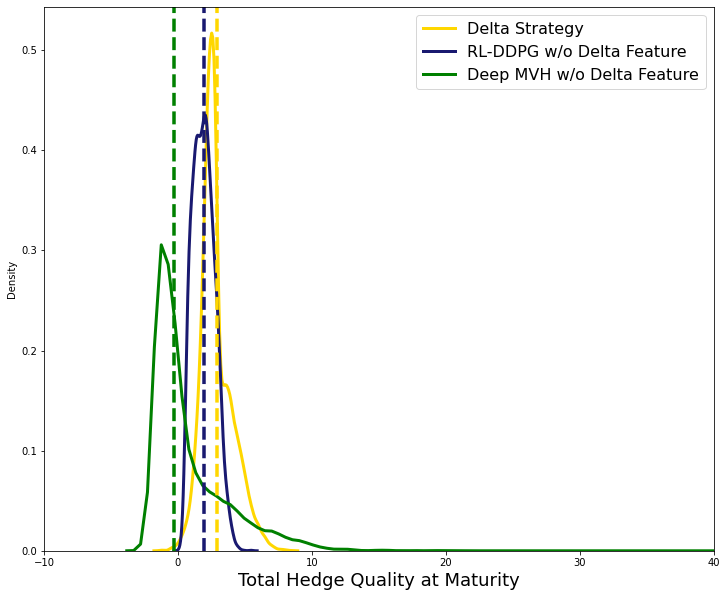}}
\caption{Histogram of total hedging cost per episode with transaction cost; Delta Hedge vs. deep-MVH (without option Delta) vs. RL-DDPG (without option Delta).  The $x$-axis is the total hedging cost during the life of the option. Positive values denote loss and negative values denote profit. Delta hedge underperforms both RL-DDPG and deep-MVH. Without option Delta, deep-MVH outperforms RL-DDPG. }\label{benchmark_rl_deep_MVH_wo_delta_feature_nonzero_TC}
\end{center}
\end{figure}

Figure \ref{SHAP_RL_deep_MVH_no_delta_nonzero_TC} summarizes how important the different state features are to RL-DDPG and 
deep-MVH agents if there are transaction costs and if option Delta is not an input. 
The ordering of features based on SHAP importance remains the same for RL-DDPG regardless whether transaction costs are absent or present. 
While option price was most important for deep-MVH without transaction costs and stock price (and other features) were much less important in that case, 
with transaction costs, stock price became the most important feature while option price now became slightly less important than the stock price. 

\begin{figure}[h]
\centering
\fbox{
\hfill
\subfigure{\includegraphics[height = 3.5 cm, width=6.8cm]{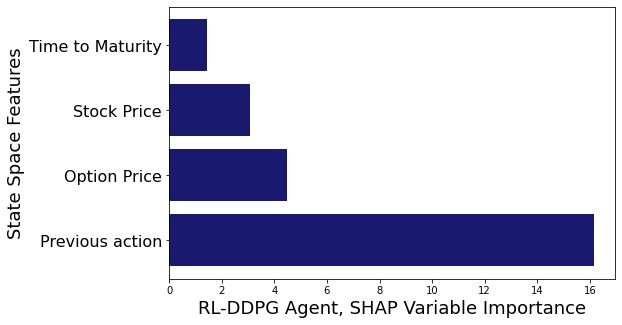}}
\hfill
\subfigure{\includegraphics[height = 3.5 cm,width=6.8cm]{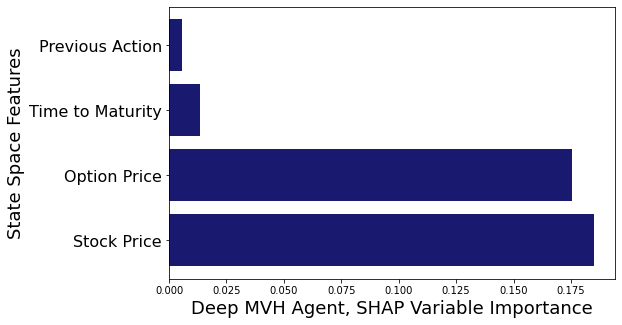}}
\hfill
}
\caption{Global SHAP variable importance for the RL-DDPG and deep-MVH without the option Delta feature under non-zero transaction cost.}\label{SHAP_RL_deep_MVH_no_delta_nonzero_TC}
\end{figure}

\subsection{Model Training}\label{model_training}
The model was trained on a computer with 8 core i7-11850 CPU $@$ 2.50GHz processor and 32.0 GB memory. 
For the RL-DDPG model, both actor and critic are trained with respect to a mean squared error (MSE) loss 
with the TensorFlow implementation of ADAM optimizer, actor learning rate and critic 
learning rate are set to $\expnumber{1}{-5}$ and $\expnumber{1}{-4}$ respectively. 
The smoothing parameter for updating the weights of the target network is $\expnumber{1}{-3}$. 
We trained the DDPG-model for one epoch consisting of 5000 training episodes. 
Following \cite{DDPG_paper}, for exploration noise, we used the Ornstein-Uhlenbeck process 
with mean-reversion and volatility parameters fixed at 0.15 and 0.2 respectively. 
The above parameters were kept fixed for both zero and non-zero transaction cost regimes.

The deep-MVH model was trained using ADAM optimization with Pytorch off-the-shelf parameters. 
Batch normalization and dropout ($p=0.25$) were applied to each layer. 
The learning rate was fixed at $\expnumber{1}{-3}$ initially and was decreased dynamically as training epochs progressed. 
The model was trained for 100 epochs and for each training epoch, a fresh batch of 50,000 training samples was simulated. 

\section{Outcome Analysis}
In this section we examine the performance of the trained RL and deep-MVH under various scenarios. 
To compare the performance of the algorithms, we studied the histogram for the total hedge 
quality at maturity (also called the hedge P\&L or hedge slippage- the wealth from the 
underlying portfolio minus the transaction cost and the final option pay-off) for various test episodes. 
Note that based on our conventions, positive values denote trader's loss. 
Note that all the testing is conducted on a sample test of 1000 episodes, drawn independently from the training sample.

\subsection{Approximating Delta Hedge under Zero-Transaction Cost}
As a first step, we perform a sanity check and look at the performance of the data driven hedging strategies in the zero-transaction cost regime. 
Our expectation is that up to the error introduced by the time discretization, the RL and deep-MVH strategies would approximate the Delta hedge. 
The performance of the three hedging strategies is shown in Figure \ref{pnl_plot_zero_TC}.
It is observed that both RL and deep-MVH lead to a distribution of the final hedge P\&L that is close to the Delta hedge one.

\begin{figure}[h]
\begin{center}
\fbox{\includegraphics[width=10cm,height=8cm]{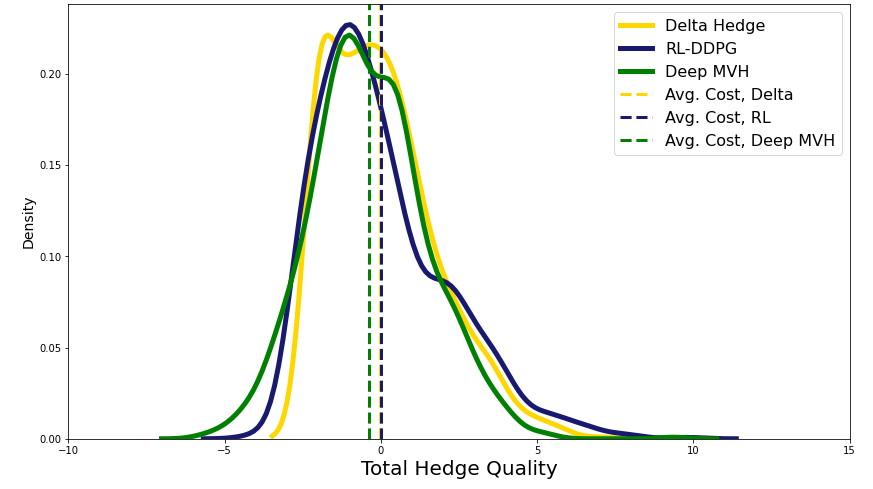}}
\caption{Hedge quality at maturity, zero transaction cost. The $x$-axis is the total hedging cost during the life of the option. Positive values denote loss and negative values denote profit.}\label{pnl_plot_zero_TC}
\end{center}
\end{figure}

\subsection{Increasing Transaction Cost}
For this test, the transaction cost parameter was increased gradually as all other parameters were kept fixed. 
The Delta hedging strategy is expected to underperform the other strategies as the transaction cost increases, 
however, we expect the data driven algorithms to learn the trade-off between the hedge quality and
transaction cost at training time and hence exhibit a more stable behavior. 
The results of the tests are shown in Figure \ref{pnl_plot_increasing_TC} below.

\begin{figure}[h]
\begin{center}
\fbox{\includegraphics[width=\textwidth,height=\textheight,keepaspectratio]{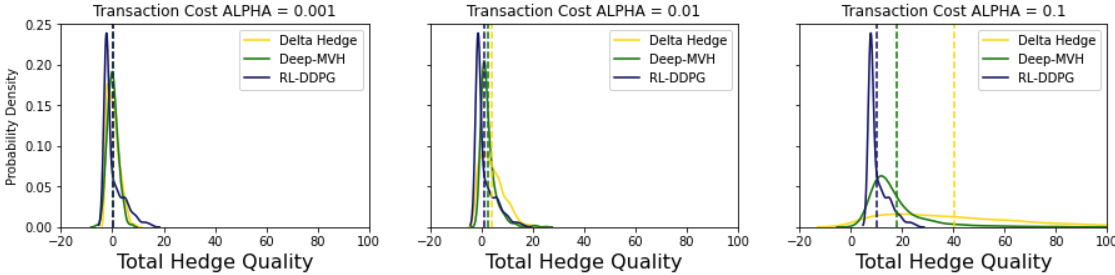}}
\caption{Hedge quality at maturity as transaction cost increases. The $x$-axis is the total hedging cost during the life of the option. Positive values denote loss and negative values denote profit.}\label{pnl_plot_increasing_TC}
\end{center}
\end{figure}

\subsection{Longer Option Maturity}\label{deep_MVH_maturity}
At a high level, hedging an option with longer maturity could be more difficult as it translates to exposure to more uncertainty.
To test this, we increased the option maturity and examined the performance of RL-DDPG and deep-MVH.

In Figure \ref{pnl_plot_longer_maturity}, it can be observed that the performance of deep-MVH model deteriorates noticeably
 as the option maturity increases. Other studies in the literature \cite{shi2021deep} reported similar behavior for other hedging problems.

This behavior may be related to the fact that the computational graph in deep-MVH becomes deeper as maturity increases 
(as opposed to the RL algorithm which trains a stationary policy that takes time as input) and this may affect how well the model can be trained.

\begin{figure}[h]
\begin{center}
\fbox{\includegraphics[width=\textwidth,height=\textheight,keepaspectratio]{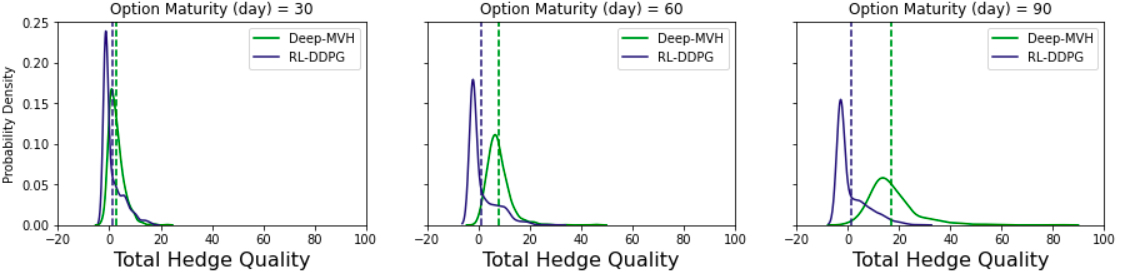}}
\caption{Hedge quality for longer option maturities. The $x$-axis is the total hedging cost during the life of the option. Positive values denote loss and negative values denote profit.}\label{pnl_plot_longer_maturity}
\end{center}
\end{figure}

\subsection{Increasing the Risk-Aversion}
To examine how the strategies trade off transaction cost with hedging P\&L variance, we consider the effect of increasing the risk aversion parameter $\lambda$.
Figure \ref{pnl_plot_increasing_lam} below demonstrates the distribution of total hedging cost for the agents. 
It is seen that the dispersion of the distributions (more noticeably for the deep-MVH agent) 
decreases as the value of the risk aversion parameter increases, as expected.  

\begin{figure}[h]
\begin{center}
\fbox{\includegraphics[width=\textwidth,height=\textheight,keepaspectratio]{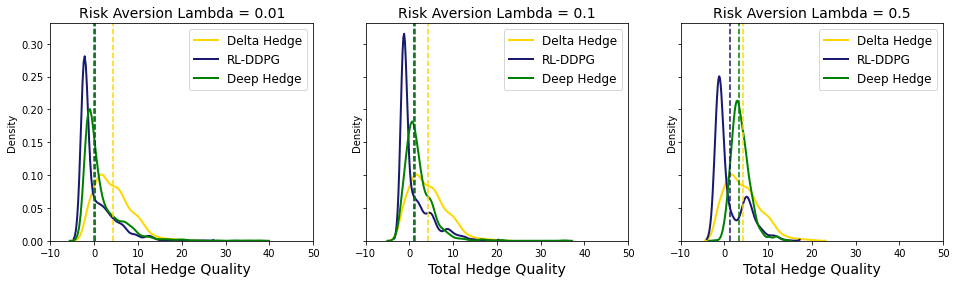}}
\caption{The distribution of total hedging cost becomes squeezed as the risk aversion parameter increases. The $x$-axis is the total hedging cost during the life of the option. Positive values denote loss and negative values denote profit.}\label{pnl_plot_increasing_lam}
\end{center}
\end{figure}

\subsection{Increasing Asset Volatility}
All things equal, more volatile markets will mean that hedge performance will be more volatile (and more dispersed) as well. 
In order to examine this effect, the volatility of the stock was increased and the performance of the hedging agents was examined. The results are presented in Figure \ref{pnl_plot_increasing_vol}. It is observed that the distributions of the realized hedging cost for all three strategies become
considerably wider as the volatility increases with the RL strategy showing a larger dispersion. Also, at high volatility levels, the mode of the distribution for RL appears to shift to the negative territory (meaning that gains from stock trades make up for transaction costs) while the deep-MVH maintains the mode around zero.

\begin{figure}[h]
\begin{center}
\fbox{\includegraphics[width=\textwidth,height=\textheight,keepaspectratio]{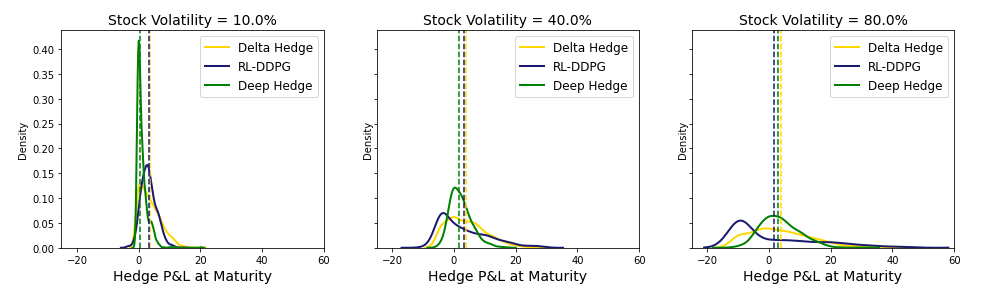}}
\caption{Hedge quality for higher asset volatility levels. The $x$-axis is the total hedging cost during the life of the option. Positive values denote loss and negative values denote profit.}\label{pnl_plot_increasing_vol}
\end{center}
\end{figure}

\section{Visualization and Interpretation}
In our experience, one of the first reactions from quants or traders, 
when exposed to the methods discussed here, is to point out their "black-box" nature. 
The trained agents generate suggested optimal actions without giving explanations. 
They do not learn parameters or parametric functions within an otherwise well-defined,
parsimonious, and understandable model. 
In this way, ML models, in particular deep neural network-based models like DDPG or deep-MVH,
achieve often state-of-the-art or at least good performance without parsimony and transparency. 
Their complexity results in increased predictive power but requires additional 
work to understand, explain, and control the agents
(A discussion of the {\em black box risk} of using ML approaches
in quantitative finance, in particular for pricing and hedging, 
can be found in \cite{cohen2021blackbox}). 

Research directions in ML such as Explainable Artificial Intelligence (XAI) and 
interpretable ML try to produce methods for either providing post-hoc human understandable explanations 
for the behavior of the model or design ML algorithms which are explainable by design 
(see \cite{Survey_ML_Explainability} for a survey). 
However, there has also been critique of using ML interpretation methods as the principal way to 
establish trust in the model (see \cite{Mythos_Interpretability}). 

For instance, one can find inconsistent definitions and motivations for interpretability in the literature
and there are different opinions about what constitutes an interpretable model. 
Also, different interpretability methods can produce several different interpretations of 
the outputs and decisions of a given model that are contradicting each other or are inconsistent with each other.

Therefore, it is not enough to produce one set of or a variety of interpretations with these interpretability 
and explainability approaches. While such approaches and methods can produce valuable insight (in particular
if several of these approaches give consistent results) or can inspire further investigations, more needs to 
be done to establish trust in the models and agents and their decisions. One still needs to assess whether 
the model is adequate and one needs to establish whether it was properly trained and tested. Generalization 
performance needs to be investigated. Model robustness and performance under generic and stressed 
conditions need to be studied. Conditions under which model performance deterioates or it becomes less
confident in its decisions must be identified and appropriate monitoring processes and safeguards need 
to be implemented. Only through this wide variety of actions can one establish trust, understanding, 
and appropriate control of such models.

\subsection{SHAP Values}\label{SHAP}
SHAP (SHapley Additive exPlanations) is a method for local interpretation of a ML model output  (\cite{lundberg2017unified}). The method attributes each individual prediction to the different predictors. For the sake of completeness, we give a high-level introduction of the concepts first.

For a ML model, the input features of the model are considered as the players of a game and the model output is the game outcome. Shapley values assign an importance score to each input (or parts of the input). For a set of $p$ input features, consider the index set $\mathcal{F} = \{1, 2, . . . , p\}$, all features in a certain coalition $S\subset\mathcal{F}$ cooperate towards the outcome $val(S)$ with the default $val(\emptyset) = 0$. Shapley values redistribute the total outcome value $val(\mathcal{F})$ among all features based on their average individual contribution across all possible coalitions $S$. 
The individual contribution of the feature $i$ in the coalition $S$ is defined by,

\begin{equation*}
\Delta(i,S) = val(S\cup\{i\})-val(S).
\end{equation*}

The Shapley value for the feature is then calculated by averaging across the coalitions,

\begin{equation*}
\phi(i) = \sum_{S\subset\mathcal{F}\setminus\{i\}}w_{S}\Delta(i,S),
\end{equation*}
where $w_{S}$ is a weighting term based on the number of choices for the subset S. 
Shapley values satisfy the properties of \textit{efficiency} ($\sum_{i\in\mathcal{F}}\phi(i)=val(\mathcal{F})$), 
\textit{Symmetry}, \textit{dummy} and \textit{additivity} (see \cite{lundberg2017unified} for details). 
The interpretation of the Shapley values for an ML model is the following. In this case, the outcome $val$ of the game
is the output of the model $f$ and Shapley values $\phi_f(i$) measure the "significance" of each feature $i$.

The SHAP framework introduced in \cite{lundberg2017unified} comes with a unifying perspective 
on connecting estimation of Shapley value (computing the Shapley value is NP-Hard) with several other popular explainability methods. 
Furthermore, they propose SHAP values as a unified measure of feature importance and prove them to be the unique solution respecting the
criteria of \textit{local accuracy}, \textit{missingness}, and \textit{consistency} (see \cite{lundberg2017unified}). 
Individual Shaply values can be averaged across the data to produce a \textit{SHAP variable importance}, therefore, for the $k$th feature, the SHAP variable importance is defined by,
\begin{equation*}
SHAPVI_{k} = \frac{1}{N}\sum_{i=1}^{N}|\phi_k^{(i)}|. 
\end{equation*}

In the following, we apply SHAP to both the RL and deep-MVH agents. In both cases, we treat the agent's policy as a 
function from the state space to action space $\pi_{\theta}:\mathcal{S}\rightarrow\mathcal{A}$. 
In that sense, we are treating the policy similar to a model in the supervised learning setting. 
We compare the trained agents in two different environments, 
zero transaction cost and non-zero transaction cost ($\alpha = 0.1$), 
with the results shown in Figure \ref{SHAP_RL}.

\begin{figure}[h]
\centering
\fbox{
\hfill
\subfigure[Zero Transaction Cost ($\alpha=0$)]{\includegraphics[height = 3.5 cm, width=6.8cm]{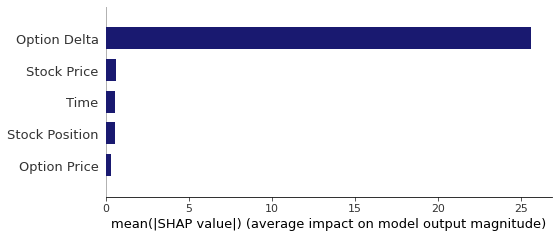}}
\hfill
\subfigure[Large Transaction Cost ($\alpha=0.1$)]{\includegraphics[height = 3.5 cm,width=6.8cm]{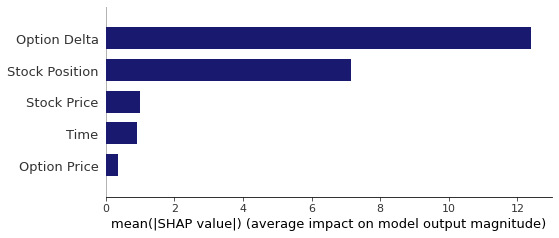}}
\hfill
}
\caption{SHAP variable importance for the RL agent in different transaction cost regimes.}\label{SHAP_RL}
\end{figure}

It is observed that in the zero-transaction cost regime, the value of the baseline (option Delta) has the highest SHAP variable importance. 
This seems reasonable, as this can be interpreted as the agent trying to follow the optimal strategy as close as possible - or as identifying
a strategy close to the Delta hedge as approximately optimal. 
In the high transaction cost regime, the variable importance scores are clearly different, 
where now the current stock position is the second most important variable. 
This can be understood as the trained RL agent exhibiting a trading "inertia" by 
identifying a \textit{no-transaction band} surrounding the current stock position 
in order to avoid too frequent transactions. Doing so will decrease the cost 
(see Figure \ref{Delta_vs_RL_daily_changes} and Figure \ref{Delta_vs_RL_1st_order_risk}).

\begin{figure}[h]
\begin{center}
\fbox{\includegraphics[width=\textwidth,height=\textheight,keepaspectratio]{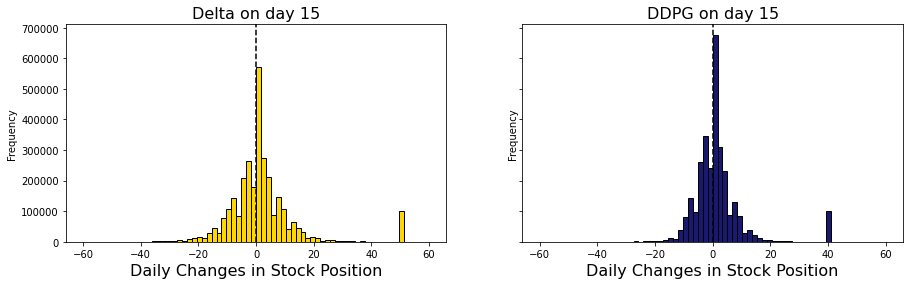}}
\caption{Histogram of portfolio rebalancing decisions made by the Delta hedge (left) and RL-DDPG (right) on day 15, in presence of transaction cost ($\alpha=0.01$). The distribution of daily changes for Delta hedge rebalancing is wider with a standard deviation of 12.4 compared to the standard deviation of 9.4 for RL-DDPG.}\label{Delta_vs_RL_daily_changes}
\end{center}
\end{figure}

\begin{figure}[h]
\begin{center}
\fbox{\includegraphics[width=\textwidth,height=\textheight,keepaspectratio]{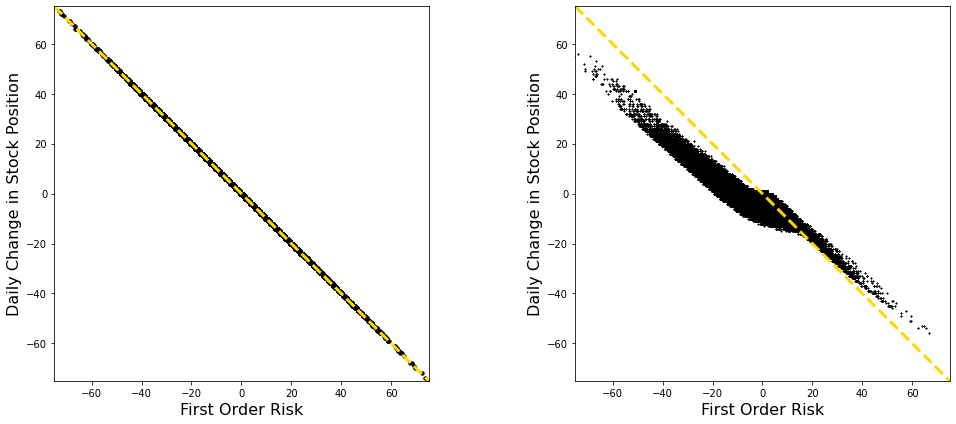}}
\caption{Scatter plot of Delta hedge rebalancing (left) and RL-DDPG (right) decisions, in presence of transaction cost ($\alpha=0.01$), based on the difference between the current stock position and option Delta (times 100). It is seen that the Delta strategy puts on the exact amount of hedge in order to stay Delta neutral. The RL-DDPG seems to have a tolerance for staying under-hedged or over-hedged (hence an inertia for rebalancing to neutralize the risk).}\label{Delta_vs_RL_1st_order_risk}
\end{center}
\end{figure}

\section{Sensitivity Analysis}
Deploying complex models in financial applications requires careful analysis of the model behavior
- and more importantly - the way it depends on various modeling choices made throughout the development process. 
More explicitly, model risk regulations (see \cite{SR11-7} for instance) require 
the model developer and the model user to perform \textit{sensitivity analysis and robustness testing}
and consider their results when deciding whether and under what circumstances the model should be used. 
This is what we will do in this section.

The notions of robustness and sensitivity can have different meanings in various modelling contexts. 
We consider sensitivity analysis as a method for understanding the relationship between input parameters and outputs 
and follow the framework developed in \cite{sedlmair2014visual} (called \textit{Parameter Space Analysis (PSA)}) 
to study methods for solving MDPs:

\textit{ “…the systematic variation of model input parameters,
generating outputs for each combination of parameters, and investigating the relation between parameter settings and corresponding
outputs”.} 

The following sections contain the results of tests designed based on our interpretation of the tasks outlined under PSA framework, 
covering \textit{training, optimization and parameter uncertainty}.

\subsection{Training Convergence}
Deep RL models are known to suffer from instabilities (that is, nonmonotonic and/or overly noisy behavior) 
when training (see \cite{RL_stability} for instance). 
During training, the cumulative reward may be unstable as training episodes progress and result in a non-monotonic curve. 
As explored in \cite{RL_stability}, generation of training data through an exploring agent and 
presence of noise in estimates of policy gradients are considered among the key reasons for this behavior. 
To best study whether such behavior occurs and impact model training and model behavior, 
several training runs under similar but different settings should be run and analyzed.

Figures \ref{DDPG_many_learning_curves} and \ref{deep_MVH_many_learning_curves} show the results 
when we repeated the training procedure with various random seeds. 
The learning curves in \ref{deep_MVH_many_learning_curves}  show clear monotonicity and convergence.
While Figure \ref{DDPG_many_learning_curves} is noisy, there is some monotonic behavior and convergence under that noise. 
These learning curves thus also show reasonable monotonicity and convergence.   

\begin{figure}[h]
\centering
\fbox{\includegraphics[width=\textwidth,height=5.0 cm,keepaspectratio]{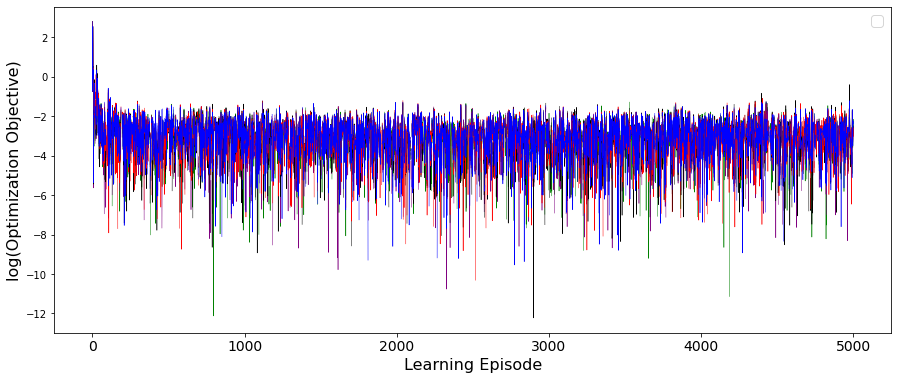}}
\caption{Plot of RL-DDPG learning curves with various random seeds.}\label{DDPG_many_learning_curves}
\end{figure}

\begin{figure}[h]
\centering
\fbox{\includegraphics[width=\textwidth,height=5.0 cm,keepaspectratio]{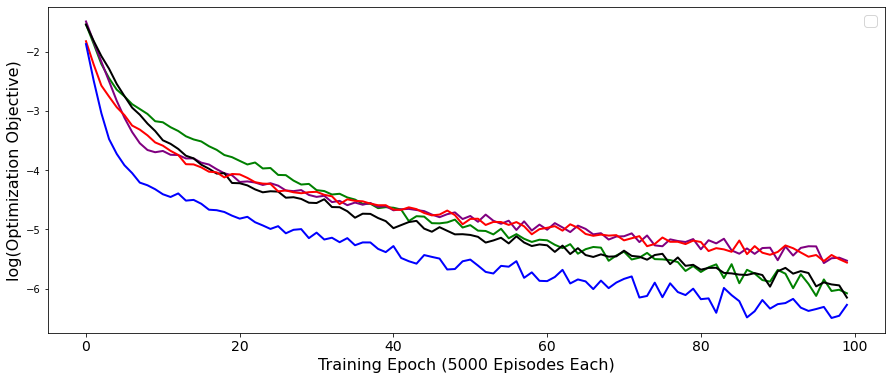}}
\caption{Plot of deep-MVH learning curves with various random seeds. The $y$-axis is the average ($\log$-)training loss for each epoch of training on 5000 episodes.}\label{deep_MVH_many_learning_curves}
\end{figure}

\begin{remark}

When comparing the training behaviors in Figures \ref{DDPG_many_learning_curves} and \ref{deep_MVH_many_learning_curves}, 
one must keep in mind that the training of the Deep MVH model proceeds in several training epochs, meaning, 
gradient updates over average loss over a batch of 5000 episodes for several times 
(also note that in each batch, a fresh set of training data is simulated). 
On the other hand, the training curves in Figure \ref{DDPG_many_learning_curves} show progress for each single episode. 
As a result, the training of the RL-DDPG is on mini-batches of $\text{size} =1$, 
whereas the Deep MVH training is on mini-batches of $\text{size} =5000$. 
This explains the noisiness observed in the RL-DDPG learning curves compared to the Deep MVH learning curves 
(see also Figure \ref{deep_MVH_many_learning_curves_mini_batch1} below).
\end{remark}

\begin{figure}[h]
\centering
\fbox{\includegraphics[width=\textwidth,height=5.0 cm,keepaspectratio]{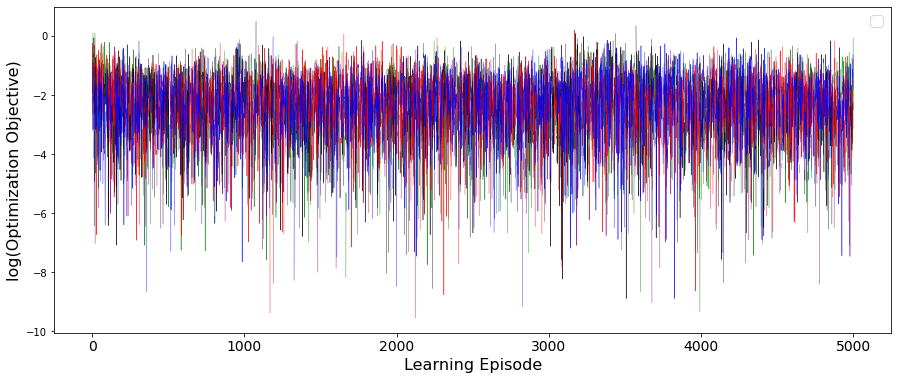}}
\caption{Plot of deep-MVH learning curves with various random seeds and mini-batch size $1$. The $y$-axis is the ($\log$-) training loss for each gradient update on a single training episode.}\label{deep_MVH_many_learning_curves_mini_batch1}
\end{figure}

\subsection{Sensitivity to Hyper-parameters}

RL agents are known to be sensitive to the architectural choices for the neural networks and the training hyper-parameters. 
Therefore, understanding the effect of various design choices on performance of the RL agents is 
an important area of research (see \cite{RL_HPO1}, \cite{RL_HPO2} for instance). 
In this section, we assess the  training and performance of both RL-DDPG and deep-MVH agents 
under various choices of hyper-parameters.

\subsubsection{Discount Factor}
As noted in \cite{RL_HPO3} for a series of benchmark tasks, the choice of reward discount factor impacts 
the performance of RL agents strongly. 
We tested the sensitivity to this parameter by varying it through the range of $[0.95, 0.97, 0.99, 0.999]$. 
For this test, we kept the transaction cost at zero and all other parameters fixed.

Figures \ref{DDPG_PnL_various_gamma} and \ref{deep_MVH_PnL_various_gamma} below 
show the total hedge quality of the RL-DDPG and deep-MVH agents under various choices of $\gamma$. 
While it seems that the changes in the deep-MVH performance are within the range of inherent variability 
and all follow a similar shape, some noticeable changes are observed in the RL-DDPG agent's total hedge quality
and the shape of the distribution changes the closer the discount factor is to one.

\begin{figure}[h]
\centering
\fbox{\includegraphics[width=10cm,height=8cm]{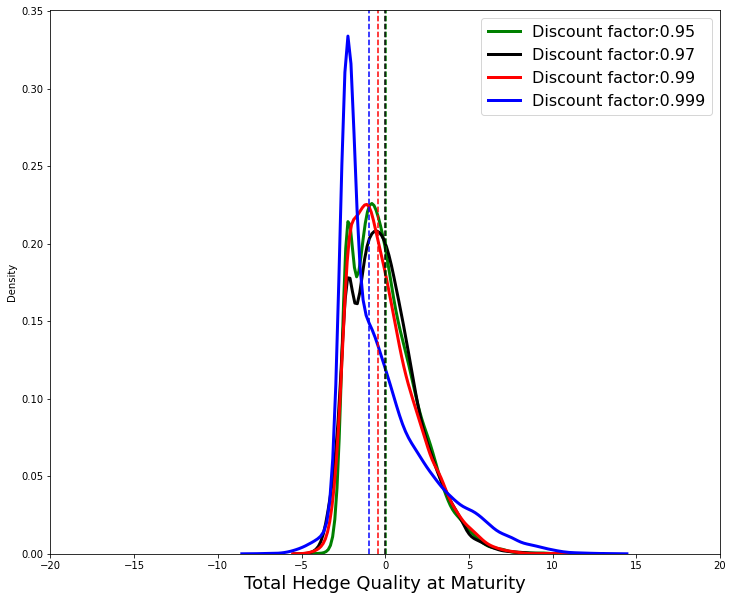}}
\caption{The histogram of total hedge quality for RL-DDPG under various reward discount factors. 
The $x$-axis is the total hedging cost during the life of the option. Positive values denote loss and negative values denote profit.}
\label{DDPG_PnL_various_gamma}
\end{figure}

\begin{figure}[h]
\centering
\fbox{\includegraphics[width=10cm,height=8cm]{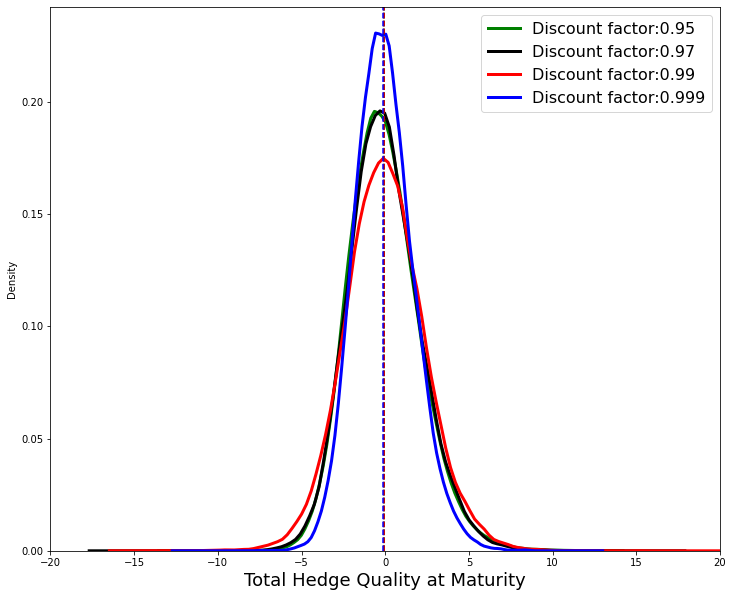}}
\caption{The histogram of total hedge quality for deep-MVH under various reward discount factors. 
The $x$-axis is the total hedging cost during the life of the option. Positive values denote loss and negative values denote profit.}
\label{deep_MVH_PnL_various_gamma}
\end{figure}

\subsubsection{Learning Rates}
In this section, we examine the effects of the learning rates in training the neural networks in deep-MVH and RL-DDPG. 
For the RL-DDPG model, we considered the common range of values for the actor-critic learning 
rates in other domains (see \cite{RL_HPO3}, \cite{RL_HPO1}, \cite{RL_HPO2}). 
More specifically, we kept the learning rate for the actor network smaller than the learning rate for the critic network. 

The hyper-parameter grid of actor-critic learning rate pairs considered is as follows: 
$\{(\expnumber{1}{-6},\expnumber{1}{-5}), (\expnumber{1}{-5},\expnumber{1}{-4}),
(\expnumber{1}{-4},\expnumber{1}{-3}), (\expnumber{1}{-3},\expnumber{1}{-2})\}$. 
The learning curves for each choice of the hyper-parameter are plotted in Figure \ref{DDPG_many_learning_curves_various_lr}. 
The pair $(\expnumber{1}{-5},\expnumber{1}{-4})$ seems to result in smaller (better) optimization objective, 
however, the differences are relatively modest 
given the inherent variability in the training algorithm (also see Figure \ref{DDPG_many_learning_curves}).

\begin{figure}[H]
\centering
\fbox{\includegraphics[width=\textwidth,height=5.0 cm,keepaspectratio]{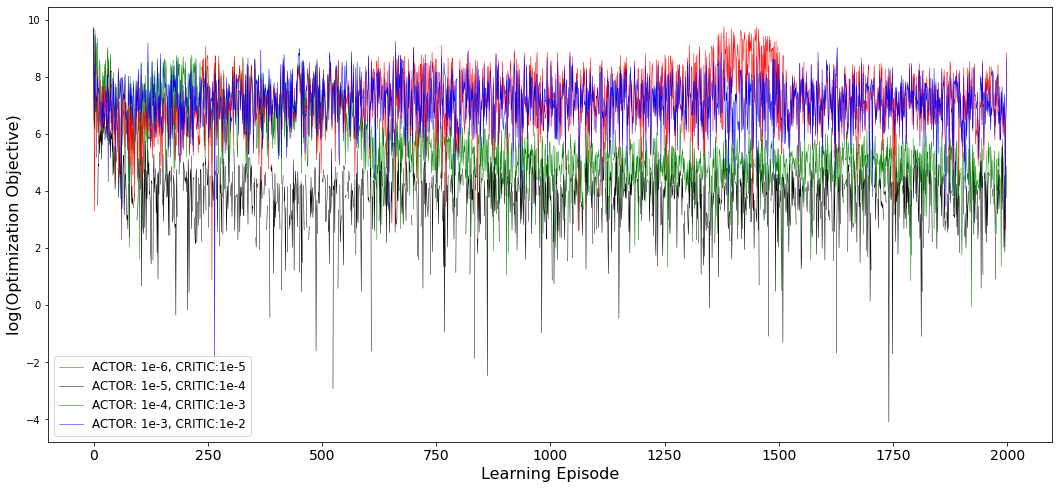}}
\caption{Plot of RL-DDPG learning curves for various actor-critic learning rates.}
\label{DDPG_many_learning_curves_various_lr}
\end{figure}

The plot in Figure \ref{deep_MVH_many_learning_curves_various_lr} below shows the learning curves for the 
deep-MVH training to various choices of learning rate in ADAM optimizer (we turned off the dynamic decreasing of the learning rate 
for this plot). 
The learning rate parameter $\expnumber{1}{-3}$ (as implemented in our default setting) 
seems to result in the best performance in the optimization for the parameters tested.

\begin{figure}[H]
\centering
\fbox{\includegraphics[width=\textwidth,height=5.0 cm,keepaspectratio]{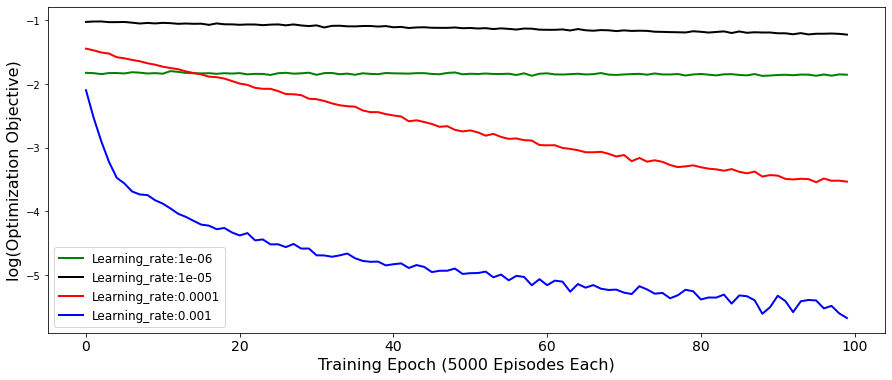}}
\caption{Plot of deep-MVH learning curves for various learning rates.}
\label{deep_MVH_many_learning_curves_various_lr}
\end{figure}

\subsubsection{Neural Network Architecture}
In this section we assess the effect of architectural choices on the training and performance of both deep MVH and RL-DDPG agents.
For the deep-MVH architecture, recall that the depth of the computational graph depends on the time discretization. 
This dependence was assessed by examining the performance for the longer maturities (see Section \ref{deep_MVH_maturity}). 
Here we examine how the width and the depths of each control at each time-step (which translates into total width and 
total depth of the global computational graph) impacts the training and the performance
of the agents. 

For this test, we considered the following feedforward architectures for each time-step $t$ policy function: 
$[5], [5,5], [10,15,10], [100,150,150,100], [100,150,300,150,100]$, where the length of each list denotes the depth of the network 
and each entry denotes the number of neurons in the corresponding layer. 
The plots in Figure \ref{deep_MVH_many_learning_curves_various_depth} and \ref{deep_MVH_PnL_various_depth} show the results. 
It is observed that the deeper architectures achieve smaller (better) optimization objective while exhibiting more noise in training. 
It is also interesting to note that  deeper architectures result in lower level of noise in total hedge P\&L at maturity.

\begin{figure}[H]
\centering
\fbox{\includegraphics[width=\textwidth,height=5.0 cm,keepaspectratio]{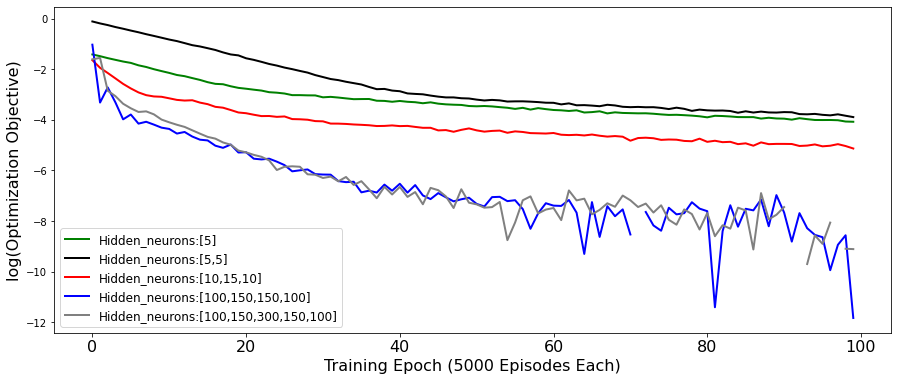}}
\caption{Plot of deep-MVH learning curves for various network depths and widths.}
\label{deep_MVH_many_learning_curves_various_depth}
\end{figure}

\begin{figure}[H]
\centering
\fbox{\includegraphics[width=10cm,height=8cm]{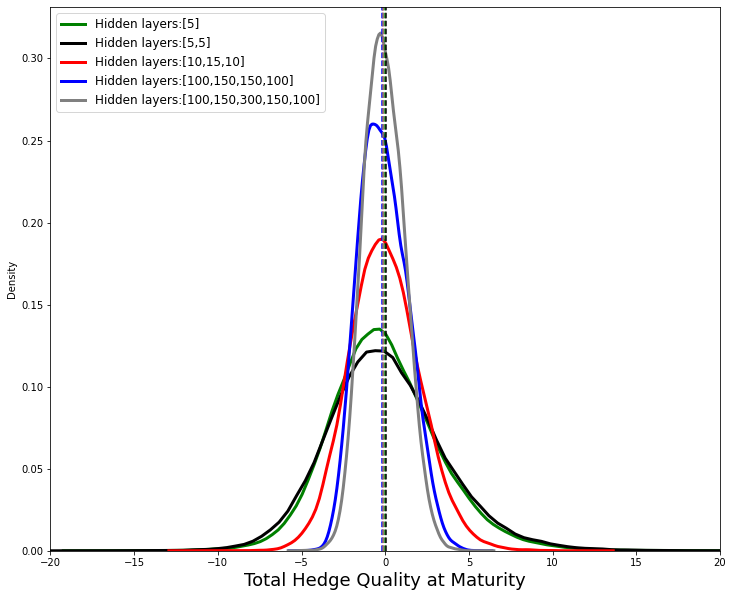}}
\caption{The histogram of total hedge quality for various depths and widths of policy networks within the deep-MVH computation graph. The $x$-axis is the total hedging cost during the life of the option. Positive values denote loss and negative values denote profit.}
\label{deep_MVH_PnL_various_depth}
\end{figure}

We repeat the same test for the actor component of  the RL-DDPG agent while keeping the critic component 
fixed and considering architectures ranging from a linear regression on the state variable to a 
deep feed forward architecture with 5 layers (see \cite{RL_HPO2}). 
The results are presented in Figures \ref{DDPG_many_learning_curves_various_depth} and \ref{DDPG_PnL_various_actor_depth}. 
The learning curves are similar across the architectures. 
Shallower networks seem to generate P\&L histograms with mode further shifted to the negative side
but the shape of the distributions is similar considering the noise in the training process and model outputs.

\begin{figure}[H]
\centering
\fbox{\includegraphics[width=\textwidth,height=8.0 cm,keepaspectratio]{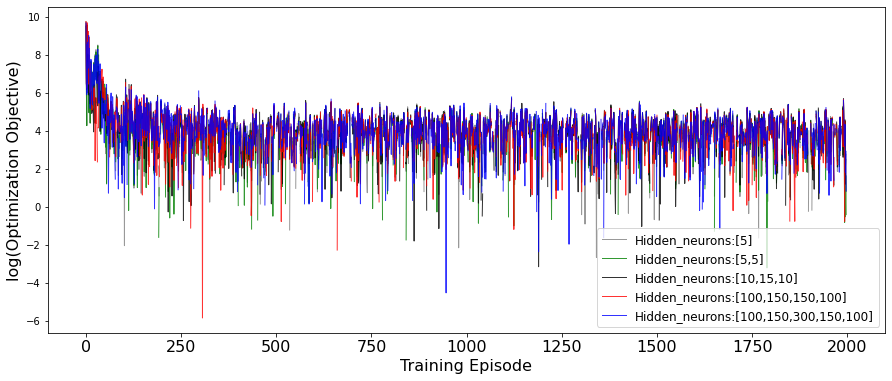}}
\caption{Plot of RL-DDPG learning curves for various depths and widths of actor neural network.}
\label{DDPG_many_learning_curves_various_depth}
\end{figure}

\begin{figure}[H]
\centering
\fbox{\includegraphics[width=10cm,height=8cm]{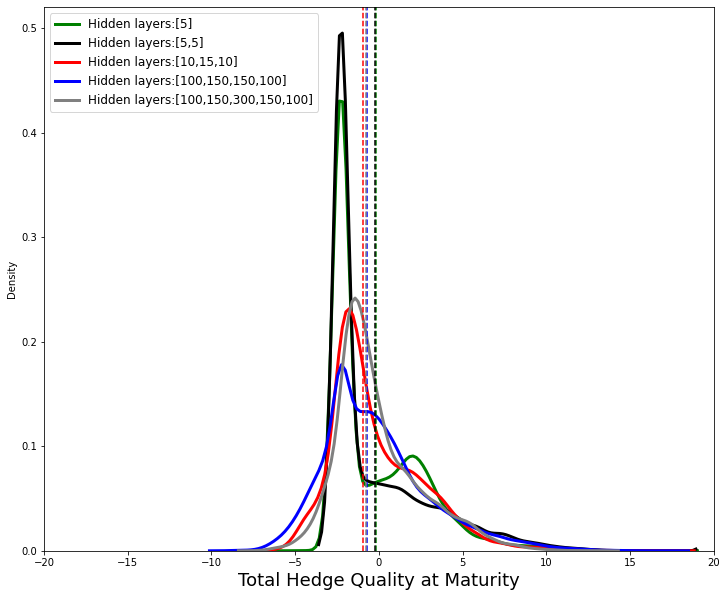}}
\caption{The histogram of total hedge quality for various depths and widths of actor network in RL-DDPG. The $x$-axis is the total hedging cost during the life of the option. Positive values denote loss and negative values denote profit.}
\label{DDPG_PnL_various_actor_depth}
\end{figure}

\subsubsection*{Neural Network Parametrization of Deep-MVH}
As mentioned in Section \ref{model_architecture}, 
we parametrized the stock position rebalancing rate by a neural network first and 
then translated the network output to the hedging decision (number of stocks to hold -- see equation (\ref{rate_to_ratio})). 
To test the sensitivity of the deep-MVH model performance to this choice of parametrization, 
we implemented an alternative architecture where at each time step, 
the position in the hedge instrument is directly parametrized by a neural network. 
Figure \ref{benchmark_rl_deep_MVH_two_parametrizations} shows that the overall performance 
of the deep-MVH model does not show significant impact to this neural network parametrization. 
However, we notice that the histogram of total hedging cost for deep-MVH with direct hedge position parametrization 
has lower dispersion.

\begin{figure}[H]
\centering
\fbox{\includegraphics[width=10cm,height=8cm]{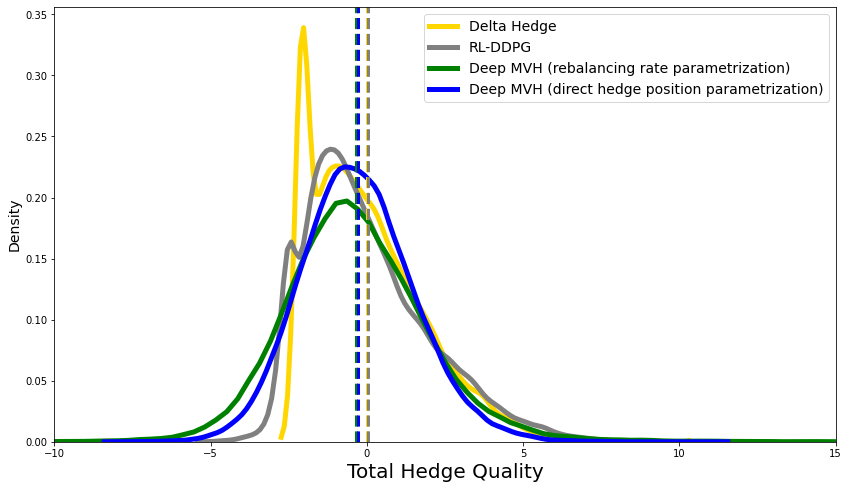}}
\caption{The histogram of total hedge quality for Delta Hedge, RL-DDPG and two different neural network parametrizations of deep-MVH (zero transaction cost).The $x$-axis is the total hedging cost during the life of the option. Positive values denote loss and negative values denote profit.}
\label{benchmark_rl_deep_MVH_two_parametrizations}
\end{figure}

\section*{Conclusions and Future Directions}
We considered reinforcement learning and trajectory-based stochastic optimal control approaches 
to stepwise mean-variance hedging in presence of transaction costs and discrete trading times. 
The main goal of the study was to describe both approaches in a self-contained manner and 
examine the performance of these two approaches on simulated data to be able to arrive at a 
clear assessment of their advantages and limitations in the simplest setting with fewest uncertainties. 
Extensive outcome analysis and sensitivity testing was performed for this purpose. 
It was observed that both RL-DDPG and deep-MVH agents approximate the 
optimal hedging strategy (Delta Hedge) when the transaction cost is zero and they outperformed 
Delta hedge as transaction cost levels increased. 
Both algorithms showed reasonable behavior in tests where the risk aversion parameter 
and market volatility increased. 
So far, the trajectory-based stochastic optimal control approach (deep MVH) does not perform well 
for longer maturities yet, but we leave the design and testing of trajectory-based approaches that perform 
better in this setting to future work. 

There are several directions to pursue in future work. One area worth studying would be 
to look at other vanilla derivatives (puts, Asians, etc.) and at portfolios of such to compare to the 
results on a single call as presented here, including the parametric case (parametric both 
in the type of instrument (strike, maturity, ...)  and in the risk aversion and objective, as done in 
\cite{deep_Bellman} for a deep-hedging type of approach for different hedge objectives). 
In a similar vein, one can look at more exotic instruments with a stronger market data dependency 
(such as Cliquets) or with embedded optionality such as Americans or Bermudans where in addition 
to hedging decisions exercise decisions need to be made, which would require an extension of the setup. 

Another important direction would be related to the models of the hedging instruments and the 
hedged instrument which here were Black-Scholes models and formulas with given, pre-calibrated 
parameters. The proposed approaches and methods should be further tested on models that capture 
more features of the market and correspond to the types of models used to simulate and price 
options in investment banks that engage in market making and competitive pricing and hedging 
of such instruments. It should be studied how one can calibrate such models best to current 
and ongoing market data and how such approaches perform when run on actually observed 
prices and quotes of hedging and hedged instruments. Instead of using calibrated or fitted models,
one could also consider the use of generative data-driven models such as various market models
to create a data-driven end-to-end setup (For an example of generative models with GANs for commodity see 
\cite{boursin2022deep}, for a variety of models for equity using GANs, normalizing flows, and/or neural SDEs
see \cite{wiese2022risk,wiese2021multi,wiese2019deep,Cohen_neural_SDE}).
 Finally, other network architectures for deep-MVH could be
explored as well as specialized training strategies that are adapted to the very deep 
computational graphs that occur in the  trajectory-based stochastic optimal control approach.

We also believe that there is still important remaining work to be done in terms of testing, analysis,
and development techniques.  For instance, a consistent and coherent way to 
analyze and explain reinforcement type and trajectory-based control agents through 
explainability techniques or setups might improve understanding of and control over such
agents and thus ease the way to further adoption of such methods (including testing and validation). 
Similarly, robustness and generalizability of such agents should be studied during 
development and deployment and continued to be monitored after deployment. This can include 
the robustness of agents trained on simulated or model data but applied to observed
real world data which connects to simulation-to-real world (\textit{Sim2Real}) transfer risk 
of RL-based autonomous agents (currently an important area of research -- see \cite{Dietterich}, \cite{RealRL}).
Adapting these methods to development of hedging and similar agents in finance would be 
an interesting direction. 

Finally, adversarial methods to find (potentially 
rare) scenarios with high risk of model failure both in simulated model settings 
but also generative models would bring another tool into the testing toolbox (see \cite{uesato2018rigorous}). 
Along these lines, \cite{cohen2021blackbox} describes many different ways how risks can materialize 
and present themselves (which they call {\em Black box risks})  
for applications of machine learning and deep learning for pricing and hedging,
but does not investigate them deeply. Deeper investigations of those risks 
would be a fruitful area for further work.

As for other formulations of  hedging objectives, we are currently studying data driven approaches 
to quadratic hedging with trajectory-based stochastic optimal control and reinforcement learning
and will discuss in a forthcoming paper. 
Applying FBSDE based methods such as deepBSDE approaches to these problems (quadratic hedging, 
stepwise mean variance hedging, global mean variance hedging) is another ongoing area of work.

\bibliographystyle{alpha}
\bibliography{rlanddsoc}

\newcommand{\etalchar}[1]{$^{#1}$}
\begin{thebibliography}{DALM{\etalchar{+}}21}

\bibitem[ARS{\etalchar{+}}20]{RL_HPO3}
Marcin Andrychowicz, Anton Raichuk, Piotr Sta{\'n}czyk, Manu Orsini, Sertan
  Girgin, Raphael Marinier, L{\'e}onard Hussenot, Matthieu Geist, Olivier
  Pietquin, Marcin Michalski, et~al.
\newblock What matters in on-policy reinforcement learning? {A} large-scale
  empirical study.
\newblock {\em arXiv preprint arXiv:2006.05990}, 2020.

\bibitem[BGTW19]{Deep_Hedging}
Hans Buehler, Lukas Gonon, Josef Teichmann, and Ben Wood.
\newblock Deep hedging.
\newblock {\em Quantitative Finance}, 19(8):1271--1291, 2019.

\bibitem[BMW22]{deep_Bellman}
Hans Buehler, Phillip Murray, and Ben Wood.
\newblock Deep {B}ellman hedging.
\newblock {\em arXiv preprint arXiv:2207.00932}, 2022.

\bibitem[BRMH22]{boursin2022deep}
Nicolas Boursin, Carl Remlinger, Joseph Mikael, and Carol~Anne Hargreaves.
\newblock Deep generators on commodity markets; application to deep hedging.
\newblock {\em arXiv preprint arXiv:2205.13942}, 2022.

\bibitem[CCHP21]{cao2021deep}
Jay Cao, Jacky Chen, John Hull, and Zissis Poulos.
\newblock Deep hedging of derivatives using reinforcement learning.
\newblock {\em The Journal of Financial Data Science}, 3(1):10--27, 2021.
\newblock Preprint version in arXiv:2103.16409.

\bibitem[CRW21]{Cohen_neural_SDE}
Samuel~N Cohen, Christoph Reisinger, and Sheng Wang.
\newblock Arbitrage-free neural-{S}{D}{E} market models.
\newblock {\em arXiv preprint arXiv:2105.11053}, 2021.

\bibitem[CSS21]{cohen2021blackbox}
Samuel~N Cohen, Derek Snow, and Lukasz Szpruch.
\newblock Black-box model risk in finance.
\newblock {\em arXiv preprint arXiv:2102.04757}, 2021.

\bibitem[DALM{\etalchar{+}}21]{RealRL}
Gabriel Dulac-Arnold, Nir Levine, Daniel~J Mankowitz, Jerry Li, Cosmin
  Paduraru, Sven Gowal, and Todd Hester.
\newblock Challenges of real-world reinforcement learning: {D}efinitions,
  benchmarks and analysis.
\newblock {\em Machine Learning}, 110(9):2419--2468, 2021.

\bibitem[Die17]{Dietterich}
Thomas~G Dietterich.
\newblock Steps toward robust artificial intelligence.
\newblock {\em Ai Magazine}, 38(3):3--24, 2017.

\bibitem[DJK{\etalchar{+}}20]{du2020deep}
Jiayi Du, Muyang Jin, Petter~N Kolm, Gordon Ritter, Yixuan Wang, and Bofei
  Zhang.
\newblock Deep reinforcement learning for option replication and hedging.
\newblock {\em The Journal of Financial Data Science}, 2(4):44--57, 2020.

\bibitem[EHJ17]{weinan2017deep}
Weinan E, Jiequn Han, and Arnulf Jentzen.
\newblock Deep learning-based numerical methods for high-dimensional parabolic
  partial differential equations and backward stochastic differential
  equations.
\newblock {\em Communications in Mathematics and Statistics}, 5(4):349--380,
  2017.
\newblock arXiv:1706.04702.

\bibitem[FWXY20]{DQL_mathematical}
Jianqing Fan, Zhaoran Wang, Yuchen Xie, and Zhuoran Yang.
\newblock A theoretical analysis of deep {Q}-learning.
\newblock In {\em Learning for Dynamics and Control}, pages 486--489. PMLR,
  2020.

\bibitem[GMR{\etalchar{+}}18]{Survey_ML_Explainability}
Riccardo Guidotti, Anna Monreale, Salvatore Ruggieri, Franco Turini, Fosca
  Giannotti, and Dino Pedreschi.
\newblock A survey of methods for explaining black box models.
\newblock {\em ACM computing surveys (CSUR)}, 51(5):1--42, 2018.

\bibitem[GYH20]{ganesan2019pricingbarriers}
Narayan Ganesan, Yajie Yu, and Bernhard Hientzsch.
\newblock Pricing barrier options with deep{B}{S}{D}{E}s.
\newblock {\em arXiv preprint arXiv:2005.10966}, May 2020.

\bibitem[GYH22]{ganesan2022pricingbarriers}
Narayan Ganesan, Yajie Yu, and Bernhard Hientzsch.
\newblock Pricing barrier options with deep backward stochastic differential
  equation methods.
\newblock {\em Journal of Computational Finance}, 25(4), 2022.

\bibitem[H{\etalchar{+}}16]{han2016deepsc}
Jiequn Han et~al.
\newblock Deep learning approximation for stochastic control problems.
\newblock {\em arXiv preprint arXiv:1611.07422}, 2016.

\bibitem[Hal20]{QLBS}
Igor Halperin.
\newblock Qlbs: {Q}-learner in the {B}lack-{S}choles(-{M}erton) worlds.
\newblock {\em The Journal of Derivatives}, 28(1):99--122, 2020.

\bibitem[Her16]{hernandez2016model}
Andres Hernandez.
\newblock Model calibration with neural networks.
\newblock {\em Available at SSRN 2812140}, 2016.

\bibitem[HIB{\etalchar{+}}18]{RL_HPO1}
Peter Henderson, Riashat Islam, Philip Bachman, Joelle Pineau, Doina Precup,
  and David Meger.
\newblock Deep reinforcement learning that matters.
\newblock In {\em Proceedings of the AAAI conference on artificial
  intelligence}, volume~32, 2018.

\bibitem[Hie19]{hientzsch2019introduction}
Bernhard Hientzsch.
\newblock Introduction to solving quant finance problems with time-stepped
  {F}{B}{S}{D}{E} and deep learning.
\newblock {\em arXiv preprint arXiv:1911.12231}, 2019.

\bibitem[Hie21]{hientzsch2021intro}
Bernhard Hientzsch.
\newblock Deep learning to solve forward-backward stochastic differential
  equations.
\newblock {\em Risk Magazine}, February 2021.

\bibitem[HXY21]{RLfinance}
Ben Hambly, Renyuan Xu, and Huining Yang.
\newblock Recent advances in reinforcement learning in finance.
\newblock {\em arXiv preprint arXiv:2112.04553}, 2021.

\bibitem[IHGP17]{RL_HPO2}
Riashat Islam, Peter Henderson, Maziar Gomrokchi, and Doina Precup.
\newblock Reproducibility of benchmarked deep reinforcement learning tasks for
  continuous control.
\newblock {\em arXiv preprint arXiv:1708.04133}, 2017.

\bibitem[KR19]{Kolm_Ritter}
Petter~N Kolm and Gordon Ritter.
\newblock Dynamic replication and hedging: {A} reinforcement learning approach.
\newblock {\em The Journal of Financial Data Science}, 1(1):159--171, 2019.

\bibitem[LHP{\etalchar{+}}15]{DDPG_paper}
Timothy~P Lillicrap, Jonathan~J Hunt, Alexander Pritzel, Nicolas Heess, Tom
  Erez, Yuval Tassa, David Silver, and Daan Wierstra.
\newblock Continuous control with deep reinforcement learning.
\newblock {\em arXiv preprint arXiv:1509.02971}, 2015.

\bibitem[Lip18]{Mythos_Interpretability}
Zachary~C Lipton.
\newblock The mythos of model interpretability: {I}n machine learning, the
  concept of interpretability is both important and slippery.
\newblock {\em Queue}, 16(3):31--57, 2018.

\bibitem[LL17]{lundberg2017unified}
Scott~M Lundberg and Su-In Lee.
\newblock A unified approach to interpreting model predictions.
\newblock {\em Advances in neural information processing systems}, 30, 2017.

\bibitem[LXL19]{liang2019deep}
Jian Liang, Zhe Xu, and Peter Li.
\newblock Deep learning-based least square forward-backward stochastic
  differential equation solver for high-dimensional derivative pricing.
\newblock {\em arXiv preprint arXiv:1907.10578}, 2019.
\newblock Also available at SSRN: https://ssrn.com/abstract=3381794 or
  http://dx.doi.org/10.2139/ssrn.3381794.

\bibitem[LXL21]{liang2021deep}
Jian Liang, Zhe Xu, and Peter Li.
\newblock Deep learning-based least squares forward-backward stochastic
  differential equation solver for high-dimensional derivative pricing.
\newblock {\em Quantitative Finance}, 21(8):1309--1323, 2021.

\bibitem[Man22]{RisknetJPM}
Rob Mannix.
\newblock {J}{P} {M}organ quants are building deep hedging 2.0.
\newblock {\em Risk.net}, 2022.

\bibitem[MK21]{mikkila2021empirical}
Oskari Mikkil{\"a} and Juho Kanniainen.
\newblock Empirical deep hedging.
\newblock {\em Available at SSRN 3923529}, 2021.

\bibitem[MKS{\etalchar{+}}15]{DQN_paper}
Volodymyr Mnih, Koray Kavukcuoglu, David Silver, Andrei~A Rusu, Joel Veness,
  Marc~G Bellemare, Alex Graves, Martin Riedmiller, Andreas~K Fidjeland, Georg
  Ostrovski, et~al.
\newblock Human-level control through deep reinforcement learning.
\newblock {\em nature}, 518(7540):529--533, 2015.

\bibitem[NIA{\etalchar{+}}18]{RL_stability}
Evgenii Nikishin, Pavel Izmailov, Ben Athiwaratkun, Dmitrii Podoprikhin, Timur
  Garipov, Pavel Shvechikov, Dmitry Vetrov, and Andrew~Gordon Wilson.
\newblock Improving stability in deep reinforcement learning with weight
  averaging.
\newblock 2018.

\bibitem[Per11]{perkowski2011backward}
Nicolas Perkowski.
\newblock Backward stochastic differential equations: {A}n introduction.
\newblock {\em Available on semanticscholar.org}, 2011.

\bibitem[Pha09]{pham2009continuous}
Huy{\^e}n Pham.
\newblock {\em Continuous-time stochastic control and optimization with
  financial applications}, volume~61.
\newblock Springer Science \& Business Media, 2009.

\bibitem[Put14]{RL_book2}
Martin~L Puterman.
\newblock {\em Markov decision processes: {D}iscrete stochastic dynamic
  programming}.
\newblock John Wiley \& Sons, 2014.

\bibitem[Res11]{SR11-7}
Federal Reserve.
\newblock Supervisory guidance on model risk management.
\newblock {\em Board of Governors of the Federal Reserve System, Office of the
  Comptroller of the Currency, SR Letter}, pages 11--7, 2011.

\bibitem[Rit17]{ritter2017machine}
Gordon Ritter.
\newblock Machine learning for trading.
\newblock {\em Available at SSRN 3015609}, 2017.

\bibitem[RSTD22]{Deep_stochastic}
A~Max Reppen, H~Mete Soner, and Valentin Tissot-Daguette.
\newblock Deep stochastic optimization in finance.
\newblock {\em arXiv preprint arXiv:2205.04604}, 2022.

\bibitem[RW19]{ruf2019neural}
Johannes Ruf and Weiguan Wang.
\newblock Neural networks for option pricing and hedging: {A} literature
  review.
\newblock {\em arXiv preprint arXiv:1911.05620}, 2019.

\bibitem[SB18]{sutton2018reinforcement}
Richard~S Sutton and Andrew~G Barto.
\newblock {\em Reinforcement learning: {A}n introduction}.
\newblock MIT press, 2018.

\bibitem[SHB{\etalchar{+}}14]{sedlmair2014visual}
Michael Sedlmair, Christoph Heinzl, Stefan Bruckner, Harald Piringer, and
  Torsten M{\"o}ller.
\newblock Visual parameter space analysis: {A} conceptual framework.
\newblock {\em IEEE Transactions on Visualization and Computer Graphics},
  20(12):2161--2170, 2014.

\bibitem[She19]{Risknet}
Nazneen Sherif.
\newblock Deep hedging and the end of the {B}lack-{S}choles era.
\newblock {\em Risk.net}, 2019.

\bibitem[SXZ21]{shi2021deep}
Xiaofei Shi, Daran Xu, and Zhanhao Zhang.
\newblock Deep learning algorithms for hedging with frictions.
\newblock {\em arXiv preprint arXiv:2111.01931}, 2021.

\bibitem[UKS{\etalchar{+}}18]{uesato2018rigorous}
Jonathan Uesato, Ananya Kumar, Csaba Szepesvari, Tom Erez, Avraham Ruderman,
  Keith Anderson, Nicolas Heess, Pushmeet Kohli, et~al.
\newblock Rigorous agent evaluation: {A}n adversarial approach to uncover
  catastrophic failures.
\newblock {\em arXiv preprint arXiv:1812.01647}, 2018.

\bibitem[WBWB19]{wiese2019deep}
Magnus Wiese, Lianjun Bai, Ben Wood, and Hans Buehler.
\newblock Deep hedging: {L}earning to simulate equity option markets.
\newblock {\em arXiv preprint arXiv:1911.01700}, 2019.

\bibitem[WHJ21]{Deep_PDE}
E~Weinan, Jiequn Han, and Arnulf Jentzen.
\newblock Algorithms for solving high dimensional {P}{D}{E}s: From nonlinear
  {M}onte {C}arlo to machine learning.
\newblock {\em Nonlinearity}, 35(1):278, 2021.

\bibitem[Wil92]{REINFORCE_paper}
Ronald~J Williams.
\newblock Simple statistical gradient-following algorithms for connectionist
  reinforcement learning.
\newblock {\em Machine learning}, 8(3):229--256, 1992.

\bibitem[WM22]{wiese2022risk}
Magnus Wiese and Phillip Murray.
\newblock Risk-neutral market simulation.
\newblock {\em arXiv preprint arXiv:2202.13996}, 2022.

\bibitem[WWP{\etalchar{+}}21]{wiese2021multi}
Magnus Wiese, Ben Wood, Alexandre Pachoud, Ralf Korn, Hans Buehler, Phillip
  Murray, and Lianjun Bai.
\newblock Multi-asset spot and option market simulation.
\newblock {\em arXiv preprint arXiv:2112.06823}, 2021.

\bibitem[YHG20]{yuhientzsch2019backward}
Yajie Yu, Bernhard Hientzsch, and Narayan Ganesan.
\newblock Backward deep{B}{S}{D}{E} methods: {A}pplications for nonlinear
  problems.
\newblock {\em arXiv preprint arXiv:2006.07635}, June 2020.
\newblock Also available at SSRN: \url{https://ssrn.com/abstract=3626208} or
  \url{http://dx.doi.org/10.2139/ssrn.3626208}.

\end{thebibliography}

\end{document}